\documentclass{lmcs}
\pdfoutput=1

\usepackage{lastpage}
\lmcsdoi{16}{4}{3}
\lmcsheading{}{\pageref{LastPage}}{}{}%
{Dec.~16,~2019}{Oct.~09,~2020}{}

\usepackage{ellipsis, mparhack, ragged2e} 
\usepackage[l2tabu, orthodox]{nag} 

\usepackage[english]{babel}
\usepackage[T1]{fontenc}
\usepackage[utf8]{inputenc}

\usepackage{amsmath, amssymb}
\usepackage[autostyle=true]{csquotes}

\PassOptionsToPackage{hyphens}{url}

\usepackage{tikz}
\usepackage{booktabs}
\usepackage{tabu,multirow}
\newcolumntype{L}{X}
\newcolumntype{R}{>{\raggedleft\arraybackslash}X}
\newcolumntype{C}{>{\centering\arraybackslash}X}

\usepackage[textsize=tiny]{todonotes}
\usepackage{algorithmicx,algorithm}
\usepackage[noend]{algpseudocode}
\usepackage{subfig}

\input{0_macros}

\begin{document}

\title[Of Cores]{Of Cores: A Partial-Exploration Framework for Markov Decision Processes}
\titlecomment{This work is an extended version of \cite{DBLP:conf/concur/KretinskyM19}, including all proofs together with further explanations and examples.
Proofs previously found in the appendix of the technical report (\texttt{v1} of this document) can be found directly after the respective Theorems~\ref{stm:core_characterization} and \ref{stm:np_complete}.
The latter proof has been corrected.
An invalid claim of \cite{DBLP:conf/concur/KretinskyM19} is corrected in Remark~\ref{rem:core_characterization_error} and the experimental evaluation is significantly extended in Section~\ref{sec:appendix:results}.
We thank the anonymous reviewers and Simon~Jantsch for their valuable feedback and spotting mistakes in the earlier versions of the manuscript.
This work has been partially supported by the Czech Science Foundation grant No.~18-11193S and the German Research Foundation (DFG) project~383882557 \enquote{Statistical Unbounded Verification} (KR~4890/2-1).}

\author[J. K{\v{r}}et\'insk\'y]{Jan K{\v{r}}et\'insk\'y}
\address{Technical University of Munich, Germany}
\email{jan.kretinsky@in.tum.de}

\author[T. Meggendorfer]{Tobias Meggendorfer}
\address{Technical University of Munich, Germany}
\email{tobias.meggendorfer@in.tum.de}

\begin{abstract}
We introduce a framework for approximate analysis of Markov decision processes (MDP) with bounded-, unbounded-, and infinite-horizon properties.
The main idea is to identify a \emph{core} of an MDP, i.e., a subsystem where we provably remain with high probability, and to avoid computation on the less relevant rest of the state space.
Although we identify the core using simulations and statistical techniques, it allows for rigorous error bounds in the analysis.
We obtain efficient analysis algorithms based on partial exploration for various settings, including the challenging case of strongly connected systems.
\end{abstract}

\maketitle
\section{Introduction} \label{sec:introduction}

Markov decision processes (MDP) are an established formalism for modelling, analysis, and optimization of probabilistic systems with non-determinism, with a large range of application domains \cite{puterman,white1993survey,white1988further}.
Classical objectives such as reachability of a~given state or the long-run average reward (mean payoff) can be solved by a variety of approaches.
In theory, the most suitable approach is linear programming as it provides exact answers (rational numbers with no representation imprecision) in polynomial time.
However, in practice for systems with more than a few thousand states, linear programming is not very usable, see, e.g., \cite{cav17}.
As an alternative, one can apply dynamic programming, typically value iteration (VI) \cite{Bellman1957}, the default method in the probabilistic model checkers PRISM \cite{PRISM} and Storm \cite{Storm}.

Despite better practical scalability of VI, systems with more than a few million states still remain out of reach of the analysis not only because of time-outs, but now also memory-outs, see, e.g., \cite{atva14}.
Surprisingly, the standard VI also suffered from a fundamental correctness issue, where convergence was only guaranteed in the limit, without a proper \emph{stopping criterion}.
Only recently, an error bound (and thus a stopping criterion) was given independently in \cite{haddad2014reachability,atva14}.
The error bound was derived from the under- and (newly obtained) over-approximations converging to the true value.
This resulted not only in error bounds on VI, but opened the door to error bounds for other techniques, including those where even convergence is not guaranteed.
For instance, while VI iteratively approximates the value of all states, the above-mentioned \emph{asynchronous VI} evaluates states at different paces.
Thus convergence is often unclear and even the rate of convergence is unknown and very hard to analyze.
However, augmenting asynchronous VI with this error bound immediately provides a correct algorithm.
A prime example is the modification of BRTDP \cite{DBLP:conf/icml/McMahanLG05} to reachability \cite{atva14} with error bounds.
These ideas are further developed for, e.g., settings with long-run average reward \cite{cav17}, continuous time \cite{atva18}, or stochastic games \cite{DBLP:conf/cav/KelmendiKKW18}.

While these solutions are efficient, they are ad-hoc, implicitly sharing the idea of \emph{simulation / learning-based partial exploration} of the system. 
In this paper, we build the foundations for designing such frameworks and provide a new perspective on these approaches, leading to algorithms for settings where previous ideas cannot apply.

In essence, the previous algorithms use (i)~simulations to explore the state space and (ii)~heuristics to analyze their experience and to bias further simulations to areas that seem more relevant for the analysis of the given property (e.g., reaching a state $s_{42}$), where (iii)~the exact VI computation takes place and yields results with a guaranteed error bound.
In contrast, this paper identifies a general concept of a \enquote{\emph{core}} of the MDP, which is independent of the particular objective (which state to reach) and, to a certain extent, even of the type of property (reachability, mean payoff, linear temporal logic formulae, etc.).
This core intuitively consists of states that are important for the analysis of the MDP, whereas the remaining parts of the state space can affect the result only negligibly.
To this end, the defining property of a core is that the system stays within the core with high probability.

There are several advantages of cores, compared to the tailored techniques.
Since the core is agnostic of any particular property, it can be \emph{re-used} for multiple queries.
Thus, the repetitive effort spent by the simulations and heuristics to explore the relevant parts of the state space by the previous algorithms can be saved.
Moreover, the general concept of cores provides a \emph{unified} understanding of the previous algorithms and allows for easier development of further partial-exploration techniques within this framework.
Additionally, identifying the core can serve to better \emph{understand} the typical behaviour of the system.
The core potentially is a lot smaller than the whole system (and thus more amenable to understand) and only contains the more likely behaviours, even for real-world models, as shown in the experimental evaluation.
In other words, the core comprises only \emph{important} states of the system.
This underlying idea of cores is not bound to MDP in any way and can be extended naturally to a broad variety of probabilistic formalisms, such as probabilistic programs, evolutionary games, and many more, providing a unified notion of importance across all of these areas.
Altogether, this motivates us to investigate this notion \emph{eo ipso}.

Moreover, in the case of MDP, making the notion of core explicit leads us to identify a new standpoint and approach for the more complicated case of strongly connected systems, where the previous algorithms as well as cores cannot help.
In technical terms, minimal cores are closed under so called \emph{end components} (parts of the state space in which the system may remain forever).
Consequently, the minimal core for a system which consists of a single end component is the whole system.
And indeed, it is impossible to give guarantees on infinite-horizon behaviour whenever a single state is ignored.
In order to provide any kind of feasible analysis for this case, we introduce the $n$-step core.
It is defined by the system staying there with high probability for $n$ steps.
This $n$-step core can naturally be used for analysis of bounded, $n$-step horizon properties.
However, by explicitly viewing the core as a set of states we are able to derive the notion of \enquote{stability} of a core.
This stability essentially describes the tendency of the probability to leave this core if longer and longer runs are considered.
We shall argue that this yields (i)~rigorous bounds for $N$-step analysis for $N \gg n$ more efficiently than a classical, direct $N$-step analysis on appropriately shaped models, and (ii)~finer information on the \enquote{long run} behaviour (for different lengths) than the summary for the infinite run, which, n.b., never occurs in reality.
This opens the door towards a rigorous analysis of \enquote{typical} behaviour of the system, with many possible applications in the design and interpretation of complex systems.

Our contribution can be summarized as follows:
\begin{itemize}
	\item
	We introduce the notion of core, study its basic properties, in its light re-interpret previous results in a unified way, and discuss its advantages.

	\item
	We stipulate a new view on long-run properties as rather corresponding to long runs than an infinite one.
	Then a modified version of cores allows for an efficient analysis of strongly connected systems, where other partial-exploration techniques necessarily fail.

	\item
	We show how these modified cores can aid in design and interpretation of systems.

	\item
	We explain how this notion can be transferred to other properties and models.

	\item
	We provide efficient, learning-based algorithms for computing both types of cores and evaluate them on several examples.
\end{itemize}
\subsection{Related Work}
Since the notion of core is fundamentally novel, we list works related to two areas of our contributions, namely (i)~to speed up (reachability) analysis of MDP and (ii)~algorithms to efficiently find small cores in practice.
Note that the former point is not a primary goal of our work, only an immediate and useful consequence.

In order to improve value iteration, several approaches are considered.
\cite{DBLP:journals/corr/abs-1910-01100} employs the idea of \emph{optimistic} value iteration, essentially guessing and verifying upper bounds, saving computational effort.
In \cite{DBLP:conf/cav/QuatmannK18}, the authors approximate the exit probability of a particular set to bound the error on the computed reachability, potentially allowing for early termination.
Note that this may seem similar to our idea, however the authors consider a fixed, pre-computed set, relative to a particular reachability query, while our approach seeks to find a set of states only dependent on the model itself.

Another natural idea is to apply state space reduction heuristics.
This includes abstraction approaches, e.g., \cite{DBLP:conf/papm/DArgenioJJL02,DBLP:conf/tacas/HahnHWZ10}, or a dual approach based on restricting the analysis to a part of the state space.
Examples of the latter are asynchronous VI in probabilistic planning, e.g., \cite{DBLP:conf/icml/McMahanLG05}, or projections in approximate dynamic programming, e.g., \cite{adp}.
In both, only a certain subset of states is considered for analysis, leading to speed ups in orders of magnitude.
However, these are best-effort solutions, which can only guarantee convergence to the true result in the limit, with no error bounds at any time.

Based on \cite{DBLP:conf/icml/McMahanLG05}, \cite{atva14} additionally provides an error bounds while only exploring a subset of states.
The approach of \cite{atva14} inspired our work and thus naturally is closely related.
In particular, our experimental evaluation shows that the approach of \cite{atva14} explores a core for almost all practical examples.
However, our fundamental goal is different:
While \cite{atva14} aims to answer a given query, we instead provide an analysis of a \emph{system}.

Algorithmically, our approach to find small cores is related to the ideas of \cite{haddad2014reachability,atva14}.
Both works maintain bounds for each state, iterating an operator similar to classical value iteration, and collapse end components to ensure convergence.
However, \cite{haddad2014reachability} constructs the whole system, in contrast to our sampling-based approach.
Our algorithms are structurally close to the BRTDP algorithm of \cite{atva14}.
We also use similar ideas to focus computation on promising areas by the means of a guided sampling approach.
However, we again emphasize that the overall goal is fundamentally different.

\section{Preliminaries} \label{sec:preliminaries}

In this section, we recall basics of probabilistic systems and set up the notation.
We assume familiarity with central ideas of measure theory.
As usual, $\Naturals$ and $\Reals$ refers to the (positive) natural numbers and real numbers, respectively.
For any set $S$, we use $\overline{S}$ to denote its complement.
A \emph{probability distribution} on a finite set $X$ is a mapping $\distribution : X \to [0,1]$, such that $\sum_{x\in X} \distribution(x) = 1$.
Its \emph{support} is denoted by $\support(\distribution) = \set{x \in X \mid \distribution(x) > 0}$.
$\Distributions(X)$ denotes the set of all probability distributions on $X$.
An event happens \emph{almost surely} (a.s.) if it happens with probability $1$.

\begin{defi}
	A \emph{Markov chain (MC)} is a tuple $\MC = (\States, \initialstate, \mctransitions)$, where
	\begin{itemize}
		\item $\States$ is a countable set of \emph{states},
		\item $\initialstate \in \States$ is the \emph{initial} state, and
		\item $\mctransitions : \States \to \Distributions(\States)$ is a \emph{transition function} that for each state $s$ yields a probability distribution over successor states.
	\end{itemize}
\end{defi}

\begin{defi}
	A \emph{Markov decision process (MDP)} is a tuple $\MDP = (\States, \initialstate, \Actions, \stateactions, \transitions)$, where
	\begin{itemize}
		\item $\States$ is a finite set of \emph{states},
		\item $\initialstate \in \States$ is the \emph{initial} state,
		\item $\Actions$ is a finite set of \emph{actions},
		\item $\stateactions: \States \to 2^{\Actions} \setminus \{\emptyset\}$ assigns to every state a non-empty set of \emph{available actions}, and
		\item $\transitions: \States \times \Actions \to \Distributions(\States)$ is a \emph{transition function} that for each state $s$ and (available) action $a \in \stateactions(s)$ yields a probability distribution over successor states.
	\end{itemize}
	We assume w.l.o.g.\ that actions are unique for each state, i.e.\ $\stateactions(s) \intersection \stateactions(s') = \emptyset$ for $s \neq s'$.
	This can be achieved by replacing $\Actions$ with $\States \times \Actions$ and adapting $\stateactions$ and $\transitions$ appropriately.
\end{defi}
For ease of notation, we overload functions mapping to distributions $f: Y \to \Distributions(X)$ by $f: Y \times X \to [0, 1]$, where $f(y, x) := f(y)(x)$.
For example, instead of $\mctransitions(s)(s')$ and $\transitions(s, a)(s')$ we write $\mctransitions(s, s')$ and $\transitions(s, a, s')$, respectively.

\begin{rem}
	In some works, Markov chains and decision processes are defined without an initial state, which instead is given as part of the query.
	While natural for some problems, our notion of cores is fundamentally dependent on the initial state.
\end{rem}

\subsection{Paths}
An \emph{infinite path} $\infinitepath$ in a Markov chain is an infinite sequence $\infinitepath = s_0 s_1 \cdots \in \States^\omega$, such that for every $i \in \Naturals$ we have that $\mctransitions(s_i, s_{i+1}) > 0$.
A \emph{finite path} (or \emph{history}) $\finitepath = s_0 s_1 \dots s_n \in \States^*$ is a finite prefix of an infinite path.
Similarly, an \emph{infinite path} in an MDP is some infinite sequence $\infinitepath = s_0 a_0 s_1 a_1 \dots \in (\States \times \Actions)^\omega$, such that for every $i \in \Naturals$, $a_i \in \stateactions(s_i)$ and $\transitions(s_i,a_i, s_{i+1}) > 0$.
\emph{Finite path}s $\finitepath$ are defined analogously as elements of $(\States \times \Actions)^* \times \States$.
We use $\infinitepath_i$ and $\finitepath_i$ to refer to the $i$-th state in the given (in)finite path.
In the following, we slightly abuse notation by identifying $(\States \times \Actions)^\omega$ and $(\States \times \Actions)^* \times \States$ with the set of infinite and finite paths, respectively.

\subsection{Strategies}
A \emph{strategy} on an MDP is a function $\strategy: (\States \times \Actions)^*\times S \to \Distributions(\Actions)$, which given a finite path $\finitepath = s_0 a_0 s_1 a_1 \dots s_n$ yields a probability distribution $\strategy(\finitepath) \in \Distributions(\stateactions(s_n))$ on the actions to be taken next.
We call a strategy \emph{memoryless randomized} (or \emph{stationary}) if it is of the form $\strategy: \States \to \Distributions(\Actions)$, and \emph{memoryless deterministic} (or \emph{positional}) if it is of the form $\strategy: \States \to \Actions$.
We denote the set of all strategies of an MDP by $\Strategies$, and the set of all memoryless deterministic strategies as $\StrategiesMD$.
Note that $\StrategiesMD$ is finite, since at each of the finitely many states there exist only finitely many actions to choose from.

Fixing any strategy $\strategy$ induces a Markov chain $\MDP^\strategy = (\States^\strategy, \initialstate^\strategy, \mctransitions^\strategy)$, where the states are given by $\States^\strategy = (\States \times \Actions)^*\times S$ and, for some state $\finitepath = s_0 a_0 \dots s_n \in \States^\strategy$, the successor distribution is defined as $\mctransitions^\strategy(\finitepath, \finitepath a_{n+1} s_{n+1}) = \strategy(\finitepath, a_{n+1}) \cdot \transitions(s_n, a_{n+1}, s_{n+1})$.

\subsection{Measures}
Any Markov chain $\MC$ induces a unique measure $\Probability<\MC>$ over infinite paths \cite[p.~758]{BaierBook}.
Assuming we fixed some MDP $\MDP$, we use $\ProbabilityStrat<\strategy><s>$ to refer to the probability measure induced by the Markov chain $\MDP^\strategy$ with initial state $s$.
See \cite[Sec.~2.1.6]{puterman} for further details.
Whenever $\strategy$ or $s$ are clear from the context, we may omit them, in particular, $\ProbabilityStrat<\strategy><>$ refers to $\ProbabilityStrat<\strategy><\initialstate>$.
For a given MDP $\MDP$ and measurable event $A$, we use the shorthand $\ProbabilityMax[A] := \sup_{\strategy \in \Strategies} \ProbabilityStrat[A]$ and $\ProbabilityMax<s>[A] := \sup_{\strategy \in \Strategies} \ProbabilityStrat<\strategy><s>[A]$ to refer to the maximal probability of $A$ over all strategies (starting in $s$).
Analogously, $\ProbabilityMin[A]$ and $\ProbabilityMin<s>[A]$ refer to the respective minimal probabilities.

Note that in general the supremum or infimum may not be obtained, depending on the structure of $A$.
However, we only consider events where an optimal strategy always exists, hence we use the superscripts $\max$ and $\min$ for emphasis.

\subsection{End components}
A non-empty set of states $C \subseteq \States$ in a Markov chain is \emph{strongly connected} if for every pair $s, s' \in C$ there is a non-trivial path from $s$ to $s'$. 
Such a set $C$ is a \emph{strongly connected component} (SCC) if it is maximal w.r.t.\ set inclusion, i.e.\ there exists no strongly connected $C'$ with $C \subsetneq C'$.
The set of SCCs 
in an MC $\MC$ is denoted by $\Sccs(\MC)$. 

The concept of SCCs is generalized to MDPs by so called \emph{(maximal) end components}.
A pair $(T, B)$, where $\emptyset \neq T \subseteq \States$ and $\emptyset \neq B \subseteq \Union_{s \in T} \stateactions(s)$, is an \emph{end component} of an MDP $\MDP$ if (i)~for all $s \in T, a \in B \intersection \stateactions(s)$ we have $\support(\transitions(s, a)) \subseteq T$, and (ii)~for all $s, s' \in T$ there is a finite path $\finitepath = s a_0 \dots a_n s' \in (T \times B)^* \times T$, i.e.\ the path stays inside $T$ and only uses actions in $B$.
Intuitively, an end component describes a set of states for which a particular strategy exists such that all possible paths remain inside these states.
By abuse of notation, we identify an end component with the respective set of states, e.g., $s \in E = (T, B)$ means $s \in T$.
An end component $(T, B)$ is a \emph{maximal end component (MEC)} if there is no other end component $(T', B')$ such that $T \subseteq T'$ and $B \subseteq B'$.
The set of MECs of an MDP $\MDP$ is denoted by $\Mecs(\MDP)$.
%
\begin{rem} \label{rem:scc_and_mec_decomposition}
	For a Markov chain $\MC$, the computation of $\Sccs(\MC)$ 
	and a topological ordering of the SCCs can be achieved in linear time w.r.t.\ the number of states and transitions by, e.g., Tarjan's algorithm~\cite{tarjan1972depth}.
	Similarly, the MEC decomposition $\Mecs(\MDP)$ of an MDP can be computed in polynomial time \cite{CY95}.
	See~\cite{CH11,ChatterjeeH12,ChatterjeeH14} for improved algorithms on general MDP and various special cases.
\end{rem}
\subsection{Objectives}
In the following, we primarily deal with unbounded and bounded variants of \emph{reachability} queries.
Essentially, for a given MDP and set of states, the task is to determine the maximal probability of reaching them, potentially within a certain number of steps.
Technically, we are interested in determining $\ProbabilityMax[\reach T]$ and $\ProbabilityMax[\reach<n> T]$, where $T$ is the set of target states and $\reach T$ ($\reach<n> T$) refers to the measurable set of runs that visit $T$ at least once (in the first $n$ steps).
The dual operators $\always T$ and $\always<n> T$ refer to the set of runs which remain inside $T$ forever or for the first $n$ steps, respectively.
See \cite[Sec.~10.1.1]{BaierBook} for further details, e.g., proofs of measurability.

Our techniques are easily extendable to other related objectives like \emph{long run average reward} (\emph{mean payoff}) \cite{puterman}, \emph{LTL formulae}, and $\omega$-regular objectives \cite{BaierBook}, or more general systems like \emph{stochastic games}.
We comment on these extensions in Section~\ref{sec:cores:extensions}.
Some of these require further knowledge about the model, which we also explain there.

\subsection{Approximate Solutions}
We are interested in finding approximate solutions efficiently, or, in other words, trading precision for speed of computation.
In our case, \enquote{approximate} means $\varepsilon$-optimal for some given precision $\varepsilon > 0$, i.e.\ the value we determine has a (guaranteed) absolute error of less than $\varepsilon$.
For example, given a reachability query $\ProbabilityMax[\reach T]$ and precision $\varepsilon$, we are interested in finding a value $v$ with $\abs{\ProbabilityMax[\reach T] - v} < \varepsilon$.

\section{The Core Idea} \label{sec:infinite_cores}

In this section, we present the novel concept of \emph{cores}, inspired by the approach of \cite{atva14}, where a specific reachability query was answered approximately through heuristic based methods.
We first establish a running example to motivate our work and explain the difference to previous approaches.

Consider the flight of an airplane. 
The plane is controlled by the pilot and the flight computer.
Together, they can take many decisions to control the plane depending on the current state.
Naturally, one may be interested in the maximal probability of arriving at the destination. 
This intuitively describes how likely it is to arrive safely, assuming that the pilot acts optimally and the computer is bug-free.
This probability may be less than $100\%$.
For example, some components may fail even under optimal conditions. 
See Figure~\ref{fig:example_full} for a simplified MDP modelling this example.

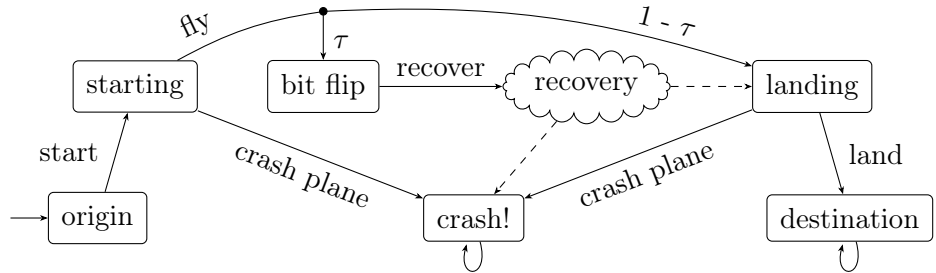
\begin{figure}[t]
	\centering
	\begin{tikzpicture}[auto,initial text=]
		\node[state,initial] at (0,0.25) (muc) {origin};
		\node[state] at (0.5,2) (start) {starting};
		\node[state] at (9.5,2) (land) {landing};
		\node[state] at (10,0.25) (prague) {destination};
		\node[state] at (3,2) (error) {bit flip};
		\node[draw,cloud,inner sep=2pt,cloud puff arc=160, cloud puffs=20, aspect=3] at (6.5,2) (recovery) {recovery};
		\node[state] at (5,0.25) (crash) {crash!};
		\node[actionnode] at (3,3) (fly-node) {};

		\path[directedge]
			(muc) edge node[anchor=east,outer sep=3pt] {start} (start)
			(land) edge node[anchor=west,outer sep=3pt] {land} (prague)
			(start) edge node[swap,sloped,anchor=north] {crash plane} (crash)
			(land) edge node[sloped,anchor=north] {crash plane} (crash)
			(error) edge node {recover} (recovery)
			(recovery) edge[dashed] (crash)
			(recovery) edge[dashed] (land)
			(crash) edge[loop below] (crash)
			(prague) edge[loop below] (prague)
		;
		\path[actionedge]
			(start) edge[bend left=10,pos=0.2] node[sloped,anchor=south] {fly} (fly-node)
		;
		\path[probedge]
			(fly-node) edge[bend left=10,pos=0.8] node[sloped,anchor=south] {1 - $\tau$} (land)
			(fly-node) edge[pos=0.6] node {$\tau$} (error)
		;
	\end{tikzpicture}

	\caption{A simplified model of a flight, where $\tau = 10^{-10}$ is the probability of potentially hazardous bit flips occurring during the flight.
	The \enquote{recovery} node represents a complex recovery procedure, comprising many states.}
	\label{fig:example_full}
\end{figure}

One key observation in \cite{atva14} is that some extreme situations may be very unlikely and we can simply assume the worst or best case for them without losing too much precision.
This allows us to completely ignore these situations.
For example, consider the unlikely event of hazardous bit flips during the flight due to cosmic radiation.
This event might eventually lead to a crash or it might have no influence on the evolution of the system at all due to redundancy.
Since this event is so unlikely to occur, we can simply assume that it always leads to a crash and still get a very precise result.
Consequently, we do not need to explore the corresponding part of the state space (the \enquote{recovery} part), saving resources.

In \cite{atva14}, the state space is explored relative to a particular reachability objective, storing upper and lower bounds on each state for the objective in consideration.
We make use of the same fundamental idea, but approach it from a different perspective, agnostic of any objective.
We are interested in finding \emph{all} relevant states of the system, i.e.\ all states which are reasonably likely to be reached.
Such a set of states is an \emph{intrinsic property} of the system, and we show that this set is both sufficient and (in a particular sense) necessary to answer reachability queries $\varepsilon$-precisely.
In particular, once computed, this set can be reused for multiple queries.

\subsection{Infinite-Horizon Cores}

First, we define the notion of an \emph{$\varepsilon$-core}.
Intuitively, an $\varepsilon$-core is a set of states which can only be exited with probability less than $\varepsilon$. 
\begin{defi}[Core]
	Let $\MDP$ be an MDP and $\varepsilon > 0$.
	A set $\Core \subseteq \States$ is an \emph{$\varepsilon$-core} if $\ProbabilityMax[\reach \overline{\Core}] < \varepsilon$, i.e.\ the probability of ever exiting $\Core$ is smaller than $\varepsilon$.
\end{defi}
When $\varepsilon$ is clear from the context, we may refer to an $\varepsilon$-core by \enquote{core}.
Observe that the core condition is equivalent to $\ProbabilityMin[\always \Core] \geq 1 - \varepsilon$, i.e.\ the probability to remain inside the core forever is large.
We highlight that the set of reachable states is a \enquote{0-core}.
As such, cores intuitively extend the straightforward idea of considering reachable states for analysis to considering only \emph{likely} reachable states.

In the following, we derive basic properties of cores, show how to efficiently construct them, and relate them to the approaches of \cite{cav17,atva14}.
First, we prove the key statement motivating our interest in cores, namely that they are both sufficient and, in a sense, required to compute $\varepsilon$-precise reachability queries.
\begin{thm} \label{stm:core_characterization}
	Let $\MDP$ be an MDP and $\varepsilon > 0$.
	A set $\Core\subseteq \States$ is an $\varepsilon$-core of $\MDP$ if and only if for every subset of states $R \subseteq \States$ we have that $0 \leq \ProbabilityMax[\reach R] - \ProbabilityMax[\reach (R \intersection \Core) \intersection \always \Core] < \varepsilon$.
\end{thm}

\begin{proof} \label{proof:core_characterization}
	We prove both directions of the equivalence separately.

	First, let $\Core$ be a core of $\MDP$ and $R \subseteq \States$ states in $\MDP$.
	Clearly,
	\begin{equation*}
		\ProbabilityMax[\reach R] \leq \ProbabilityMax[\reach R \intersection \reach \overline{\Core}] + \ProbabilityMax[\reach R \intersection \always \Core]
	\end{equation*}
	by simple case distinction.
	Furthermore, we have that
	\begin{equation*}
		\ProbabilityMax[\reach R \intersection \always \Core] = \ProbabilityMax[\reach (R \intersection \Core) \intersection \always \Core] \leq \ProbabilityMax[\reach R]
	\end{equation*}
	and
	\begin{equation*}
		0 \leq \ProbabilityMax[\reach R \intersection \reach \overline{\Core}] \leq \ProbabilityMax[\reach \overline{\Core}] < \varepsilon.
	\end{equation*}
	Together, we obtain
	\begin{multline*}
		0 \leq \ProbabilityMax[\reach R] - \ProbabilityMax[\reach (R \intersection \Core) \intersection \always \Core] \leq \\ \ProbabilityMax[\reach R \intersection \reach \overline{\Core}] + \ProbabilityMax[\reach R \intersection \always \Core] - \ProbabilityMax[\reach (R \intersection \Core) \intersection \always \Core] < \varepsilon.
	\end{multline*}
	For the other direction, assume that $\States' \subsetneq \States$ is not a core.
	Now, pick $R = \overline{\States'}$, $R \neq \emptyset$ by assumption.
	Clearly, $R \intersection S' = \emptyset$, hence we only need to prove that $\ProbabilityMax[\reach R] > \varepsilon$.
	By definition, since $\States'$ is not a core, we have that $\ProbabilityMax[\reach \overline{\States'}] > \varepsilon$.
\end{proof}
This theorem shows that for any reachability objective $R$, we can determine $\ProbabilityMax[\reach R]$ up to $\varepsilon$ precision by determining the reachability of $R$ on the sub-model induced by any $\varepsilon$-core, i.e.\ by only considering runs which remain inside $\Core$.
Conversely, the theorem also shows that if we would consider a set of states not satisfying the core property then there is at least one reachability property that we cannot answer with epsilon-precision guarantees.

\begin{figure}
	\centering
	\begin{tikzpicture}[auto,initial text=]
		\node[state,initial above] at (0,0) (s0) {$s_0$};
		\node[actionnode] at (1.25,0) (a01) {};
		\node[actionnode] at (-1.25,0) (a02) {};
		\node[state] at (3,-0.75) (s1) {$\underline{s_1}$};
		\node[state] at (3,0.75) (s2) {$s_2$};
		\node[state] at (-3,0.75) (s3) {$s_3$};
		\node[draw,cloud,inner sep=2pt,cloud puff arc=160, cloud puffs=20, aspect=3] at (-3.6,-0.8) (cloud) {unknown};
	
		\path[directedge]
			(s1) edge[loop right] (s1)
			(s2) edge[loop right] (s2)
			(s3) edge[loop left] (s3)
		;
		\path[actionedge]
			(s0) edge node[action] {$a$} (a01)
			(s0) edge[swap] node[action] {$b$} (a02)
		;
		\path[probedge]
			(a01) edge[swap] node[prob] {$0.8$} (s1)
			(a01) edge node[prob] {$0.2$} (s2)
			(a02) edge[swap] node[prob] {$0.3$} (s3)
			(a02) edge node[prob] {$0.7$} (cloud)
		;
		\path[directedge]
			(cloud) edge[gray,anchor=center] node[prob,fill=white,inner sep=2pt] {?}(s1)
		;
	\end{tikzpicture}
	
	\caption{An MDP showing that cores are not always required to answer reachability queries $\varepsilon$-precisely.}
	\label{fig:core_larger_than_required}
\end{figure}
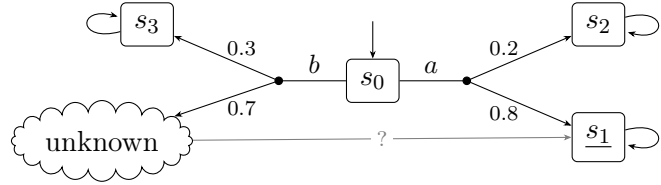
\begin{rem} \label{rem:core_characterization_error}
	In the conference paper \cite{DBLP:conf/concur/KretinskyM19} we incorrectly reported a stronger statement, claiming that cores are required to compute \emph{any} property (except some corner cases).
	This mistake was discovered independently by the authors and reviewers of this work.
	We show a counterexample to this claim in Figure~\ref{fig:core_larger_than_required}.
	Here, we can already see that the maximal probability of reaching state $s_1$ is $0.8$ by choosing action $a$ in the initial state $s_0$, independent of the system's behaviour in the \enquote{unknown} area, as we explain below.
	However, in order to obtain an $\varepsilon$-core for any $\varepsilon < 0.7$, we would need to explore further.

	Now, we explain the example of Figure~\ref{fig:core_larger_than_required} in more detail.
	Action $b$ immediately leads us to state $s_3$ with $0.3$ probability, a MEC with neither outgoing edges nor a target state.
	The probability of reaching the target $s_1$ after choosing action $b$ thus is at most $0.7$, independent of the probability of reaching $s_1$ from the unknown region, indicated by the grey arrow.
	In other words, one can derive the upper bound $0.7$ on the probability of reaching the target after taking action $b$ without investigating the unknown area.
	Dually, following action $a$ yields a lower bound of $0.8$, hence it is clear that the probability of reaching the target is at least $0.8$.
	Moreover, since the remaining probability of $0.2$ after taking action $a$ also leads to such an \enquote{absorbing} MEC, we can conclude that the maximal probability of reaching $s_1$ is $0.8$.
	Note that if instead of $s_2$ there would be another unknown region, we would need to explore it, since it may also eventually lead to the target set.

	We emphasize that the counterexample relies on this particular structure of the MDP relative to the reachability objective, i.e.\ that there is a \enquote{shortcut} to the goal as offered by action $a$ together with an immediate dead-end associated with all other actions (see action $b$ in this example).
	In our experiments, we only rarely observed such a structure.
\end{rem}
Theorem~\ref{stm:core_characterization} motivates us to find cores.
Of course, one could simply construct the whole state set $\States$, since it is a core for any $\varepsilon$.
Note that, in a sense, this is what traditional explicit methods are doing.
However, this clearly does not yield any computational advantages.
Thus, we naturally are interested in finding a core which is as small as possible, which we call a \emph{minimal core}.
\begin{defi}[Minimal Core]
	Let $\MDP$ be an MDP and $\varepsilon > 0$.
	$\MinimalCore \subseteq \States$ is a \emph{minimal $\varepsilon$-core} if it is minimal w.r.t.\ set inclusion, i.e.\ $\MinimalCore$ is an $\varepsilon$-core and there exists no $\varepsilon$-core $\Core' \subsetneq \MinimalCore$.
\end{defi}
When $\varepsilon$ is clear from the context, we may refer to a minimal $\varepsilon$-core by \enquote{minimal core}.
In the running example, a minimal core for $\varepsilon = 10^{-6}$ would contain all states except the \enquote{bit flipped} state and the \enquote{recovery} subsystem, as they are reached with probability $\tau \ll \varepsilon$.

Unfortunately, finding small cores for $\varepsilon > 0$ is computationally quite expensive, as we show in the following.
We first prove that there may be several minimal cores for one system.
While this statement is rather obvious, we include it due to the instructiveness of its proof, hinting at the underlying combinatorics we use in the following proof.
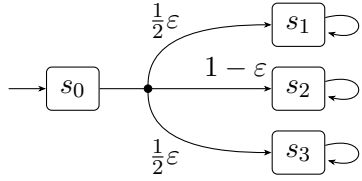
\begin{figure}[t]
	\centering
	\begin{tikzpicture}[auto,initial text=]
	\node[state,initial] at (0,0) (s0) {$s_0$};
	\node[actionnode] at (1,0) (action) {};
	\node[state] at (3,0.85) (s1) {$s_1$};
	\node[state] at (3,0) (s2) {$s_2$};
	\node[state] at (3,-0.85) (s3) {$s_3$};
	
	\path[directedge]
	(s1) edge[loop right] (s1)
	(s2) edge[loop right] (s3)
	(s3) edge[loop right] (s3)
	;
	\path[actionedge]
	(s0) edge (action)
	;
	\path[probedge]
	(action) edge[out=90,in=180,pos=0.3] node[anchor=south] {$\tfrac{1}{2} \varepsilon$} (s1)
	(action) edge[pos=0.7] node {$1 - \varepsilon$} (s2)
	(action) edge[out=-90,in=180,pos=0.3] node[anchor=north] {$\tfrac{1}{2} \varepsilon$} (s3)
	;
	\end{tikzpicture}
	
	\caption{A simple MDP showing that minimal cores are not unique.}
	\label{fig:core_non_unique}
\end{figure}
\begin{prop}[Non-uniqueness] \label{stm:non_unique}
	There is an MDP with minimal cores $\MinimalCore, \MinimalCore'$ satisfying $\MinimalCore \neq \MinimalCore'$ for any $0 < \varepsilon < \frac{1}{2}$.
\end{prop}
\begin{proof}
	Consider the MDP shown in Figure~\ref{fig:core_non_unique}.
	Any $\varepsilon$-core contains the states $s_0$ and $s_2$.
	But $\{s_0, s_2\}$ is not a valid core, since $\ProbabilityMax[\reach \overline{\{s_0, s_2\}}] = \ProbabilityMax[\reach \{s_1, s_3\}] = \varepsilon$.
	Hence, at least one of $s_1$ and $s_3$ has to be part of a core.
	It is easy to verify that both $\{s_0, s_1, s_2\}$ and $\{s_0, s_2, s_3\}$ are (minimal) cores.
\end{proof}
By extending the above example we can show that there indeed might be exponentially many minimal cores.
More importantly, we observe that finding a core of a given size (for a non-trivial $\varepsilon$) is NP-complete.
\begin{figure}[t]
%
%
%
%
%
%
%
%
%
%
	\centering
	\begin{tikzpicture}[auto,initial text=]
		\node[state,initial above] at (0,-0.5) (s0) {$s_0$};

		\node[state] at (-3,-2) (e1) {$\{v_1, v_2\}$};
		\node[state] at (-1,-2) (e2) {$\{v_2, v_4\}$};
		\node at (1,-2) (edots) {$\dots$};
		\node[state] at (3,-2) (em) {$\{v_3, v_n\}$};

		\node[state] at (-2.5,-4) (v1) {$v_1$};
		\node[state] at (-1,-4) (v2) {$v_2$};
		\node at (1,-4) (vdots) {$\dots$};
		\node[state] at (2.5,-4) (vn) {$v_n$};

		\node[state] at (0,-5.5) (sink) {$s_-$};

		\node[actionnode] (e1a) at (-3,-2.75) {};
		\node[actionnode] (e2a) at (-1,-2.75) {};
		\node[actionnode] (edotsa) at (1,-2.75) {};
		\node[actionnode] (ema) at (3,-2.75) {};

		\path[directedge]
			(s0) edge (e1)
			(s0) edge (e2)
			(s0) edge[dashed] (edots)
			(s0) edge (em)

			(v1) edge (sink)
			(v2) edge (sink)
			(vdots) edge[dashed] (sink)
			(vn) edge (sink)

			(sink) edge[loop below] node {$1$} (sink)
		;

		\path[actionedge]
			(e1) edge (e1a)
			(e2) edge (e2a)
			(edots) edge[dashed] (edotsa)
			(em) edge (ema)
		;

		\path[probedge,densely dotted]
			(e1a) edge[swap] (v1)
			(e1a) edge (v2)
			(e2a) edge (v2)
			(e2a) edge (vdots)
			(ema) edge (vdots)
			(ema) edge (vn)
		;
		\path[probedge,dashed]
			(edotsa) edge (vdots)
		;
	\end{tikzpicture}
	
	\caption{The MDP used in the reduction from Vertex cover to cores.
		All dotted transitions (from \enquote{edge} state actions to \enquote{vertex} states) have a transition probability of $\frac{\varepsilon}{2}$.
		For readability, edges from \enquote{edge} state actions to the sink state, carrying the remaining $1 - \varepsilon$ probability, are omitted.
	}
	\label{fig:core_np}
\end{figure}
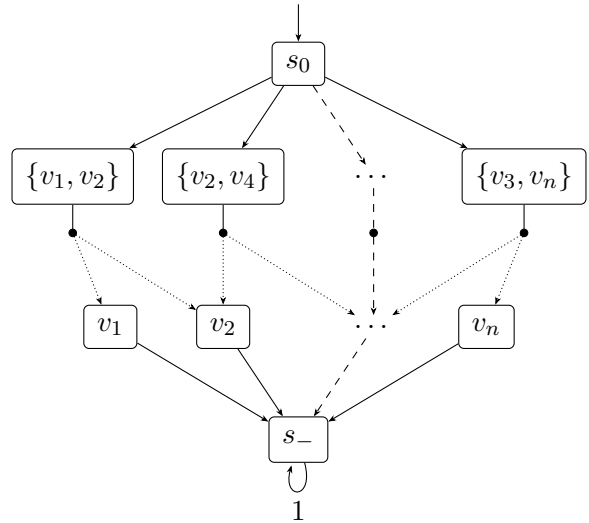
\begin{thm}[NP-completeness] \label{stm:np_complete}
	The problem $\{(\MDP, k) \mid \text{$\MDP$ has an $\varepsilon$-core of size $k$}\}$ is NP-complete for any $0 < \varepsilon < 1$.
\end{thm}
\begin{proof} \label{proof:np_complete}
	\textit{Containment:} The problem is in NP, since the reachability problem of a given set of states in MDP is in P.
	Thus, a core serves as its own, linearly sized certificate.

	\textit{Hardness:} For hardness, we show a reduction from the \texttt{VERTEX-COVER} problem.
	We briefly recall this problem: One is given an (undirected) graph $(V, E)$ and a threshold $k \in \Naturals$.
	The question is whether it is possible to find a subset of vertices $V' \subseteq V$ with size at most $k$, i.e.\ $\cardinality{V'} \leq k$, such that each edge of the graph is incident to at least one vertex in $V'$, i.e.\ for every $\{v, w\} \in E$ either $v \in V'$ or $w \in V'$ (or both).

	Let thus a graph $(V, E)$ and number $k \in \Naturals$ be an instance of \texttt{VERTEX-COVER}.
	We construct the MDP $\MDP$ as depicted in Figure~\ref{fig:core_np}.
	We define the set of states $\States = \{s_0, s_-\} \union V \union E$, i.e.\ there is a state for each edge and each vertex of the graph together with two special states $s_0$ and $s_-$.
	In $s_0$, there is an action for each edge, leading to the respective edge's state with probability 1.
	Each edge state has a single available action, leading to the two vertices incident to this edge with probability $\frac{\varepsilon}{2}$ each.
	The remaining probability of $1 - \varepsilon$ directly leads to the sink state $s_-$.
	Each vertex state directly moves to the sink state with probability 1.
	Essentially, in the initial state, we can pick any edge of the graph and probabilistically move to its two vertices.
	Clearly, the size of the MDP is linear in the size of the graph (note that $\varepsilon$ is fixed and thus not part of the input).

	We show that this MDP admits a core of size at most $2 + \cardinality{E} + k$ iff the input graph has a vertex cover of size at most $k$.
	First, assume that a vertex cover $V'$ exists.
	We show that $\Core = \{s_0, s_-\} \union E \union V'$ is a core by contradiction.
	Assume that the probability of reaching $\overline{\Core} = V \setminus V'$ is at least $\varepsilon$ and let $\strategy$ be a deterministic witness strategy.
	Let $e = \{v, w\} \in E$ the edge which is selected by $\strategy$ in the initial state, i.e.\ the unique state with $\transitions(s_0, \strategy(s_0), e) = 1$.
	Clearly, neither $v$ nor $w$ are in $\Core$, since otherwise $\strategy$ would not be a witness.
	This contradicts the assumption that $V'$ is a vertex cover, since $e$ is not covered by it.

	Now, assume that a core $\Core$ of size $2 + \cardinality{E} + k$ exists.
	Clearly, $s_0$, $s_-$ and all edge states are part of any core, since all of them can be reached with probability one (and $\varepsilon < 1$).
	Hence, let $V' = V \intersection \Core$ be the vertex states in the core.
	Clearly, $\cardinality{V'} \leq k$.
	We show that $V'$ is a vertex cover.
	Let $e = \{v, w\} \in E$ be any edge.
	By analogous argumentation as above we necessarily have either $v \in V'$ or $w \in V'$.
\end{proof}
Observe that this result only implies that finding the smallest minimal cores is hard.
By virtue of Theorem~\ref{stm:core_characterization}, we still are interested in finding small, not necessarily minimal cores---any reduction in number of states directly translates to a speed-up in subsequent computations.
Thus we introduce a learning-based approach which quickly identifies reasonably sized cores in the following section.

\begin{rem} \label{rem:polytime_preprocessing}
	We highlight that the above proof shows NP-completeness even for the restricted class of acyclic MDP.
	However, an NP-completeness proof for Markov chains remains open.
	Instead of using actions to choose edge states, we could introduce a uniform distribution over them.
	However, then the transition probabilities from edge states to vertex states would need to be re-weighted to $\cardinality{E} \frac{\varepsilon}{2}$.
	Consequently, $\varepsilon$ has to be smaller than $\frac{1}{\cardinality{E}}$ and thus this does not prove the NP-hardness for arbitrary $\varepsilon$.

	An interesting variant is the computation of a \enquote{pre-core}: states which necessarily are contained in any core for some given $\varepsilon$, as, for example, $s_0$, $E$ and $s_-$ in the above hardness proof.
	We suspect that a greedy algorithm may be able to identify such states in PTIME.
	Such an analysis may yield an interesting preprocessing step to quickly identify very important states, thus speeding up computation of cores.
	This particularly may be helpful for systems which are strongly connected, which we explain later.
\end{rem}

\subsection{Learning a Core}

\newcommand{\upperbound}{U}

\begin{algorithm}[t]
	\newcommand{\pop}{pop}
	\caption{\textsc{LearnCore}} \label{alg:learn_infinite_core}
	\begin{algorithmic}[1]
		\Require MDP $\MDP$, precision $\varepsilon > 0$, upper bounds $\upperbound$, state set $\Core$ with $\initialstate \in \Core$
		\Ensure $\Core$ s.t. $\Core$ is an $\varepsilon$-core
		\While{$\upperbound(\initialstate) \geq \varepsilon$}
			\State $\finitepath \gets \textsc{SamplePath}(\initialstate, \upperbound)$ \Comment Generate path
			\State $\Core \gets \Core \union \finitepath$ \Comment Expand core
			\State $\textsc{UpdateECs}(\Core, \upperbound)$
			\For{$s$ in $\finitepath$} \Comment Back-propagate values
				\State $\upperbound(s) \gets \min\{\upperbound(s), \max_{a \in \Actions(s)} \sum_{s' \in S} \Delta(s, a, s') \cdot \upperbound(s')\}$ \label{line:learn_infinite_core:update}
			\EndFor
		\EndWhile
		\State \Return $\Core$
	\end{algorithmic}
\end{algorithm}

As we have shown in the previous section, finding a minimal core is NP-complete, hence we aim for a best-effort, learning-based algorithm, often identifying a small core.
To this end, we introduce a guided, sampling-based algorithm in Algorithm~\ref{alg:learn_infinite_core}.
Our method is structurally very similar to the algorithm introduced in \cite{atva14}.
Nevertheless, we present it explicitly here since (i)~it is significantly simpler and (ii)~we introduce modifications later on.
We assume that the model is described by an initial state and a successor function, yielding all possible actions and the resulting distribution over successor states, instead of an explicit list of transitions.
This allows us to only construct a small fraction of the state space and achieve sub-linear runtime (in the number of states and transitions) for many models.
In particular, we observe in Section~\ref{sec:experiments} that for some models we are able to identify a small core orders of magnitude faster than the construction of the state set $\States$, speeding up subsequent computations drastically.

During the execution of the algorithm, the system is traversed by following the successor function, starting from the initial state.
Each state encountered is stored in a set of \emph{explored} states, all other, not yet visited states are \emph{unexplored}.
Unexplored successors of explored states are called \emph{partially explored}:
The algorithm is aware of their existence but has no other information about these states.
Furthermore, the algorithm stores for each (explored) state $s$ an upper bound $\upperbound(s)$ on the probability of reaching unexplored states starting from $s$.
The algorithm gradually grows the set of explored states and simultaneously updates their upper bounds safely until the desired threshold is achieved in the initial state, i.e.\ $\upperbound(\initialstate) < \varepsilon$.
Then, the set of explored states provably satisfies the core property.
In particular, the upper bound is updated by sampling a path according to \textsc{SamplePath} and back-propagating the values along that path using Bellman updates (also called Bellman backups).

\textsc{SamplePath} samples paths following some heuristic.
These paths do not have to be rooted in the initial state $\initialstate$, follow the transition probabilities given by the successor function, resolve non-determinism in a particular way, or be of a particular length.
For example, a successor might be sampled with probability proportional to its upper bound times the transition probability or in a round-robin scheme\footnote{For example, by numbering the successors of some action arbitrarily, we simply select a successor in ascending fashion whenever we choose that particular action.}.
In our implementation, we follow the former idea.
The intuition is as follows.
We want to explore states which indeed are likely to be reached, since those have to be included in a core anyway.
But we do not want to waste computational effort on states which have a small probability of reaching new unexplored states.
The product of transition probability and upper bound is only large if that successor is both likely to be reached and has a (presumably) large chance of reaching a new unexplored state.
Otherwise, the successor probably is hardly reachable in general or we already have gathered enough information and hence do not need to explore further in that direction.
As we show later in the experimental evaluation in Section~\ref{sec:experiments}, using the currently stored upper bounds as guidance often yields significant speed-ups in practice.
\textsc{SamplePath} can also incorporate machine learning techniques and domain knowledge about the system, yielding even better suggestions about important states.

\textsc{UpdateECs} identifies MECs of the currently explored sub-system and \enquote{collapses} them into a single representative state.
Alternatively, this can be viewed as linking the bounds of all states in each end component together.
In particular, each state's bound is set to the maximum bound of all actions leaving the end component, omitting all \enquote{internal} actions.
This is necessary to ensure convergence of the upper bounds to the correct value.
Technically this process removes spurious fixed points of $\upperbound$.
We briefly explain this issue in the following, it is more thoroughly explained in, e.g., \cite{atva14,cav17}.

Recall that from each state within an EC we can reach every other state of the EC with probability one.
Remaining inside the (explored) EC will not lead to an unexplored state.
If, for example, some state $s$ can reach unexplored states with probability $0.5$, so can every state $s'$ in the EC by first moving to $s$ and then following the actions necessary to achieve the $0.5$ probability.
Setting the upper bound of all states in an EC to the maximum upper bound of all \enquote{outgoing} actions thus intuitively is correct---but it is also necessary for convergence:
Observe that by definition, the upper bound of each state initially is set to $1$.
Now, consider, for example, a MEC consisting of a single state $s$ with a self loop under action $a$.
Since $s$ can reach a state with upper bound $1$ under action $a$ (namely itself), the update in Line~\ref{line:learn_infinite_core:update} of the algorithm will always keep $\upperbound(s)$ at $1$.
By identifying $(\{s\}, \{a\})$ as a MEC and removing the internal action, we ensure convergence.
Furthermore, MECs without outgoing edges are the only parts of the system which \enquote{create} $0$ upper bounds---only there do we know for sure that no unexplored state can be reached.
We omit a precise definition of the underlying MEC-quotienting procedure \cite{de1997formal}, since it entails a lot of technical subtleties, distracting from our main result.
For the sake of understanding the algorithm, it is safe to assume that $\MDP$ does not contain any MECs except trivial sinks, which we can easily identify and immediately assign an upper bound of $0$.

For (a.s.) termination, we only require that the sampling heuristic is \enquote{(almost surely) fair}.
This means that (i)~any partially explored state is reached eventually (a.s.), in order to explore a sufficient part of the state space, and (ii)~any explored state with $U(s) > 0$ is visited infinitely often (a.s.), in order to back-propagate values accordingly.
Observe that we do not require that $\finitepath$ always starts in $\initialstate$, only that this happens again and again.
Further, we require that the initial upper bounds are consistent with the given state set, i.e.\ $U(s) \geq \ProbabilityMax<s>[\reach \overline{\Core}]$.
This is trivially satisfied by $U(\cdot) = 1$.
Note that in contrast to \cite{atva14}, the set whose reachability we approximate dynamically changes and, further, only upper bounds are computed.

\begin{thm}
	Algorithm~\ref{alg:learn_infinite_core} is correct and terminates (a.s.) if \textsc{SamplePath} is (a.s.) fair and the given upper bounds $U$ are consistent with the given set $\Core$.
\end{thm}
\begin{proof}
	\emph{Correctness}: By assumption $U(s)$ initially is a correct upper bound for the \enquote{escape} probability, i.e.\ $U(s) \geq \ProbabilityMax<s>[\reach \overline{\Core}]$.
	Each update in Line~\ref{line:learn_infinite_core:update} preserves correctness, independent of the sampled path.
	Moreover, we set $\upperbound(s) \gets 0$ if \textsc{UpdateECs} identifies an EC without outgoing actions, which trivially is correct, too.
	Hence, if $\upperbound(\initialstate) < \varepsilon$, we have $\ProbabilityMax<s>[\reach \overline{\Core}] < \varepsilon$.

	\emph{Termination}: Recall that we assumed that \textsc{SamplePath} is (a.s.) fair.
	This implies that eventually the whole model will be explored, i.e.\ $\Core = \States$ (otherwise there would exist a partially explored state which is never visited, contradicting the fairness condition).
	Consequently, all MECs will be collapsed by \textsc{UpdateECs}.
	In particular, all MECs without outgoing actions have their upper bound $\upperbound$ set to $0$.
	Moreover, since $\upperbound$ is monotonically decreasing by definition, the upper bounds of any state $s$ converge to a value $\upperbound^*(s)$.
	Now, assume that there exists a state $s$ where $\upperbound^*(s) > 0$, i.e.\ its upper bound does not converge to zero.
	Let w.l.o.g.\ $s$ be a state with maximal $\upperbound^*$ among all states.
	Recall that every state is visited infinitely often by our fairness assumption, in particular $s$.
	By definition of $\upperbound$, it is easy to see that all successors $s'$ of $s$ under any action necessarily have the same value $\upperbound^*(s') = \upperbound^*(s)$, since otherwise the value of $s$ would eventually be decreased by an update.
	Now, this implies that the set of states with maximal values, i.e.\ $\{s' \mid \upperbound^*(s') = \upperbound^*(s)\}$ is closed under the transition dynamics of the system and contains at least one end component, contradicting the fact all ECs are collapsed by \textsc{UpdateECs}.
	Consequently, we have $\upperbound^*(s) = 0$ for any state $s$, in particular we have that eventually $\upperbound(\initialstate) < \varepsilon$, proving termination.
\end{proof}
As Algorithm~\ref{alg:learn_infinite_core} is correct and terminates for any faithful upper bounds and initial state set, we can restart the algorithm and interleave it with other approaches refining the upper bounds.
For example, one could periodically update the upper bounds using, e.g., strategy iteration, which can speed up convergence drastically for particular models.
Further, we can reuse the computed upper bounds and state set to compute a core for a tighter precision.

\subsection{Extending Cores to other Properties and Models} \label{sec:cores:extensions}

We explain how a core can be used for verification and how our approach differs from existing ones.
Of course, we can compute reachability or safety objectives on a given core $\varepsilon$-precisely.
In this case, our approach conceptually is not too different from the one in \cite{atva14}.
Yet, we argue that our approach yields a stronger result.
Due to cores being an intrinsic object, we are able to reuse and adapt this idea easily to many other objectives.
Observe that a dedicated adaption may still yield slightly better performance, but requires significantly more work.
For example, see \cite{cav17} for an adaption to mean payoff.

To see how we can connect our idea to mean payoff, we briefly explain this objective and then recall an observation of \cite{cav17}.
First, rational rewards are assigned to each state, which are obtained upon each visit to that state.
Then, the mean payoff of a particular run is the limit average reward obtained from the visited states.
The mean payoff under a particular strategy then is obtained by integrating over the set of all runs.
As mentioned by \cite{cav17}, a mean payoff objective can be decomposed into a separate analysis of each (explored) MEC and a (weighted) reachability query
\begin{equation*}
	\textsf{optimal mean payoff} = \sup_{\strategy \in \Strategies} {\sum}_{M \in \Mecs(\MDP)} \textsf{mean payoff of $\strategy$ in $M$} \cdot \ProbabilityStrat<\strategy>[\reach \always M].
\end{equation*}
Since we can bound the reachability on unexplored MECs by the core property, we can easily bound the error on the computed mean payoff (assuming we know an a-priori lower and upper bound on the reward function).
Consequently, we can approximate the optimal mean payoff by only analysing the corresponding core.

Similarly, LTL queries and parity objectives can be answered by a decomposition into analysis of MECs and their reachability.
Intuitively, given a MEC one can decide whether the MEC is \enquote{winning} or \enquote{losing} for these objectives.
The overall probability of satisfying the objective then equals the probability of reaching a winning MEC \cite{BaierBook}.
Again, we can bound the reachability of unexplored MECs and thus the error we incur when only analysing the core.
Note that the statement of Theorem~\ref{stm:core_characterization} directly carries over to these settings.
Moreover, it also transfers to \emph{minimal reachability / satisfaction} queries.
Intuitively, the minimal probability of reaching a given set is obtained by maximizing the probability of reaching states from which the given set can be avoided forever, i.e.\ ECs not intersecting the given set.
More formally, we have $\ProbabilityMin<\MDP, s>[\reach R] = 1 - \ProbabilityMax<\MDP, s>[\always \overline{R}]$.
Observe that in this sense minimal reachability can be phrased as maximizing satisfaction of an LTL query.
More directly, we can again decompose minimal reachability into a graph analysis to identify all \enquote{safe} ECs and then apply maximal reachability.
Hence, given a core, we can approximate the minimal probability up to a precision of $\varepsilon$ and a variant of Theorem~\ref{stm:core_characterization} is applicable.

In general, many verification tasks can be decomposed into a reachability query and analysis of specific parts of the system.
Since our framework is agnostic of the verification task in question, it can be transparently plugged in to obtain significant speed-ups at a \emph{controllable} loss of precision.

We highlight that our approach moreover is directly applicable to models with infinite state space, since finite cores still may exist for these models.
Moreover, we can also apply our core idea directly to \emph{stochastic games}, i.e.\ MDP where an additional, antagonistic player is introduced.
Here, we can compute a core by interpreting the game as an MDP where both players cooperate.
In other words, a core of a stochastic game is a set of states where neither player can ever escape from with significant probability.
It is not difficult to see that the essence of Theorem~\ref{stm:core_characterization} also carries over to this setting.

Even more generally, the essential idea of cores, namely to classify states as \enquote{important} based on the probability of them occurring along a path can be transferred to many other probabilistic formalisms, immediately providing an intuitive, unified notion of importance to these areas.
For example, the concept of \emph{stochastic invariants} \cite{DBLP:conf/popl/ChatterjeeNZ17} of probabilistic programs is equivalent to a core of the underlying Markov chain.

\section{Beyond Infinity} \label{sec:finite_cores}

In the previous section, we have seen that MECs play an essential role for many objectives.
Hence, we study the interplay between cores and MECs.
\begin{prop} \label{stm:reachable_mecs_in_core}
	Let $M \in \Mecs(\MDP)$ be a MEC.
	If there is a state $s \in M$ with $\ProbabilityMax[\reach \{s\}] \geq \varepsilon$ then $M \subseteq \Core$ for every $\varepsilon$-core $\Core$.
\end{prop}
\begin{proof}
	Recall that for $s, s' \in M$, we have $\ProbabilityMax<s>[\reach \{s'\}] = 1$, thus $\ProbabilityMax[\reach \{s\}] = \ProbabilityMax[\reach \{s'\}] \geq \varepsilon$ and thus $s' \in \Core$.
\end{proof}
This implies that sufficiently reachable MECs always need to be contained in a core \emph{entirely}.
Many models comprise only a few or even a single MEC, e.g., restarting protocols like mutual exclusion or biochemical models of reversible reactions.
Together with the result of Theorem~\ref{stm:core_characterization}, i.e.\ constructing a core is necessary for $\varepsilon$-precise answers, this shows that in general we cannot hope for any reduction in state space, even when only requiring $\varepsilon$-optimal solutions. 
In our experimental evaluation, strongly connected components prove to be a challenge for our approach, since it spends a lot of time computing unnecessary information until eventually the whole model is explored.
Here, a \enquote{pre-core} approach as mentioned in Remark~\ref{rem:polytime_preprocessing} may help to improve performance.

However, real-world models often exhibit a particular structure, with many states only being visited infrequently.
For example, a biological process may reach some ratio of specimen very rarely.
Since we necessarily have to give up on something to obtain further savings, we propose an extension of our idea, motivated by a modification of our running example.

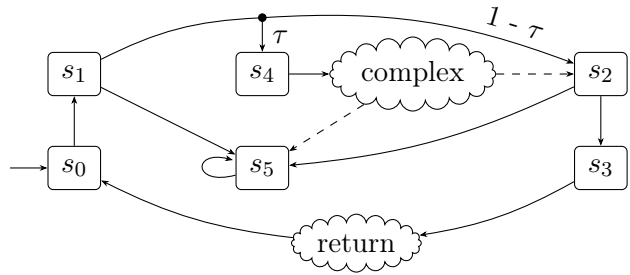
\begin{figure}[t]
	\centering
	\begin{tikzpicture}[auto,initial text=]
		\node[state,initial] at (0,0) (muc) {$\initialstate$};
		\node[state] at (0,1.25) (start) {$s_1$};
		\node[state] at (7,1.25) (land) {$s_2$};
		\node[state] at (7,0) (prague) {$s_3$};
		\node[state] at (2.5,1.25) (error) {$s_4$};
		\node[draw,cloud,inner sep=2pt,cloud puff arc=160, cloud puffs=20, aspect=3] at (4.5,1.25) (recovery) {complex};
		\node[state] at (2.5,0) (crash) {$s_5$};
		\node[actionnode] at (2.5,2) (fly-node) {};
		\node[draw,cloud,inner sep=2pt,cloud puff arc=160, cloud puffs=20, aspect=3] at (3.75,-1) (return) {return};

		\path[directedge]
			(muc) edge node[sloped,anchor=south] {} (start)
			(land) edge node[sloped,anchor=south] {} (prague)
			(start) edge node[swap,sloped,anchor=north] {} (crash)
			(land) edge[bend left=10] node[sloped,anchor=north] {} (crash)
			(error) edge node {} (recovery)
			(recovery) edge[dashed] (crash)
			(recovery) edge[dashed] (land)
			(prague) edge[bend left=10] (return)
			(return) edge[bend left=10] (muc)
			(crash) edge[loop left] (crash)
		;
		\path[actionedge]
			(start) edge[bend left=12] node[sloped,anchor=south] {} (fly-node)
		;
		\path[probedge]
			(fly-node) edge[bend left=12,pos=0.8] node[sloped,anchor=south] {1 - $\tau$} (land)
			(fly-node) edge node {$\tau$} (error)
		;
	\end{tikzpicture}

	\caption{An adaptation of the model from Figure~\ref{fig:example_full}, with an added return trip, represented by the \enquote{return} node.
	State and action labels have been omitted for readability.}
	\label{fig:example_loop}
\end{figure}

Instead of a one-way trip, consider the plane going back and forth, as shown in Figure~\ref{fig:example_loop}.
Now, the plane eventually will suffer from a bit flip.
Additionally, assuming that there is a non-zero probability of not being able to recover from the error, the plane will eventually crash with probability 1, \emph{independent} of the strategy.

We make two observations.
First, any core needs to contain at least parts of the recovery sub-system, since it is reached with probability $1$.
Thus, this (complex) sub-system has to be constructed.
Second, the witness strategy is meaningless, since any strategy is optimal---the crash cannot be avoided in the long run.
In particular, deliberately crashing the plane has the same long run performance as flying it \enquote{optimally}.
Note that this is quite different from computing the optimal strategy for a single trip and applying it repeatedly.
In practice, we are, in fact, often interested in the performance of such a model for a long, but not necessarily infinite, horizon.

To this end, one could compute the step-bounded variants of the objectives, but this incurs several problems: (i)~choosing a sensible step bound $n$, (ii)~computational overhead (a precise computation has worst-case complexity of $\cardinality{\Delta} \cdot n$ even for reachability), and (iii)~all states reachable within $n$ steps have to be constructed (which equals the whole state space for practically all models and reasonable choices of $n$).
In the following, we present a different approach to this problem, again based on the idea of cores.

\subsection{Finite-Horizon Cores}

We introduce \emph{finite-horizon cores}, which are completely analogous to (infinite-horizon) cores, only with a step bound attached to them.

\begin{defi}[Finite-Horizon Core]
	Let $\MDP$ be an MDP, $\varepsilon > 0$, and $n \in \Naturals$.
	A set $\FiniteCore \subseteq \States$ is an \emph{$n$-step $\varepsilon$-core} if $\ProbabilityMax[\reach<n> \overline{\FiniteCore}] < \varepsilon$
	and it is a \emph{minimal $n$-step $\varepsilon$-core} if it additionally is minimal w.r.t.\ set inclusion.
\end{defi}
As before, whenever $n$ or $\varepsilon$ are clear from the context, we may drop the corresponding part of the name.
Similar properties hold and we omit the completely analogous proof.
\begin{thm}
	Let $\MDP$ be an MDP, $\varepsilon > 0$, and $n \in \Naturals$.
	Then $\FiniteCore \subseteq \States$ is an $n$-step $\varepsilon$-core if and only if for all $R \subseteq \States$ we have $0 \leq \ProbabilityMax[\reach<n> R] - \ProbabilityMax[\reach<n> (R \intersection \FiniteCore) \intersection \always<n> \FiniteCore] < \varepsilon$. 
\end{thm}
Finite-horizon cores are much smaller than their \enquote{infinite} counterparts on some models, even for large step bounds $n$.
For instance, in our modified running example of Figure~\ref{fig:example_loop}, omitting the \enquote{complex} states gives an $n$-step core even for very large $n$ (depending on $\tau$).
On the other hand, finding such finite cores seems to be harder in practice.
Naively, one could apply the core learning approach of Algorithm~\ref{alg:learn_infinite_core} to a modified model where the number of steps is encoded into the state space, i.e.\ $\States' = \States \times \{0, \dots, n\}$.
However, this comes with a huge increase in space complexity, since we store and back-propagate $\cardinality{\States} \cdot n$ values instead of only $\cardinality{\States}$.
Nevertheless, we can efficiently approximate them by enhancing our previous approach with further observations.

\subsection{Learning a Finite Core} \label{sec:finite_cores:stability}

\begin{algorithm}[t]
	\newcommand{\pop}{pop}
	\newcommand{\len}{len}
	\caption{\textsc{LearnFiniteCore}} \label{alg:learn_finite_core}
	\begin{algorithmic}[1]
		\Require MDP $\MDP$, precision $\varepsilon > 0$, step bound $n$, upper bounds \textsc{GetBound} / \textsc{UpdateBound}, state set $\FiniteCore$ with $\initialstate \in \FiniteCore$
		\Ensure $\FiniteCore$ s.t. $\FiniteCore $ is an $n$-step $\varepsilon$-core
		\While{$\textsc{GetBound}(\initialstate, n) \geq \varepsilon$} \label{alg:learn_finite_core:loop}
			\State $\finitepath \gets \textsc{SamplePath}(\initialstate, n, \textsc{GetBound})$ \Comment Generate path of length $n$ 
			\State $\FiniteCore \gets \FiniteCore \union \finitepath$ \Comment Update Core
			\For{$i \in [n-1, n-2, \dots, 0]$} \Comment Back-propagate values
				\State $s \gets \finitepath_i$, $r \gets n - i$
				\State $\textsc{UpdateBound}\left(s, r, \max_{a \in \Actions(s)} \sum_{s' \in S} \Delta(s, a, s') \cdot \textsc{GetBound}(s', r - 1) \right)$\label{alg:learn_finite_core:update}
			\EndFor
		\EndWhile
		\State \Return $\FiniteCore$
	\end{algorithmic}
\end{algorithm}

In Algorithm~\ref{alg:learn_finite_core}, we present our learning variant for the finite-horizon case.
This algorithm is structurally very similar to the previous Algorithm~\ref{alg:learn_infinite_core}.
The fundamental difference is in Line~\ref{alg:learn_finite_core:update}, where the bounds are updated.
One key observation is that the probability of reaching some set $R$ within $k$ steps is at least as high as reaching it within $k - 1$ steps, i.e.\ $\ProbabilityMax<s>[\reach<k> R] < \varepsilon$ is non-decreasing in $k$ for any $s$ and $R \subseteq \States$.
Therefore, we can use function over-approximations to store upper bounds sparsely and avoid storing $n$ values for each state.
To allow for multiple implementations, we thus delegate the storage of upper bounds to an abstract function approximation, namely \textsc{GetBound} and \textsc{UpdateBound}.
This approximation scheme is supposed to store and retrieve the upper bound of reaching unexplored states for each state and number of \emph{remaining} steps.
We only require it to give a \emph{consistent} upper bound, i.e.\ whenever we call $\textsc{UpdateBound}(s, r, p)$, $\textsc{GetBound}(s, r')$ will return at least $p$ for all $r' \geq r$.
Moreover, we require the trivial result $\textsc{GetBound}(s, 0) = 0$ for all states $s$.
In Section~\ref{sec:finite_core:approx}, we list several possible instantiations.
\begin{thm}
	Algorithm~\ref{alg:learn_finite_core} is correct if \textsc{UpdateBound} and \textsc{GetBound} are consistent and correct w.r.t.\ the given state set $\FiniteCore$.
	Further, if \textsc{UpdateBound} stores all values precisely and \textsc{SamplePath} yields any path of length $n$ infinitely often (a.s.), the algorithm terminates (a.s.).
\end{thm}
\begin{proof}
	\emph{Correctness}:
	As before, the upper bound function is only updated through Bellman backups, which preserve correctness.

	\emph{Termination}:
	Given that the upper bound function stores all values precisely, the algorithm is an instance of asynchronous value iteration, which is guaranteed to converge~\cite{puterman}.
	More formally, observe that in this case we can essentially characterize \textsc{UpdateBound} and \textsc{GetBound} by a value vector $v : \States \times \{0, \dots, n\} \to [0, 1]$.
	Then, the update in Line~\ref{alg:learn_finite_core:update} corresponds to setting $v(s, r) = \max_{a \in \stateactions(s)} \sum_{s' \in \States} \transitions(s, a, s') \cdot v(s', r - 1)$ and $v(s, 0) = 0$.
	Assuming that every possible path of length $n$ is sampled infinitely often, we update every possible state-step pair infinitely often.
	Thus, assume for contradiction that there is a state $s$ and step $r$ where $v(s, r) > 0$ but there exists a path of length $n$ containing state $s$ at position $n - r$.
	W.l.o.g.\ assume that $r$ is minimal among all such state-step pairs with $v(s, r) > 0$.
	Clearly, $r > 0$ by definition of $v$.
	But, since $r$ is minimal, we have that $v(s', r - 1) = 0$ for all reachable $s'$.
	Since there exists a path reaching $s$, all of its successors are reachable and hence have a value of $0$.
	Consequently, the algorithm eventually updates $v(s, r) = 0$, contradicting the assumption.
\end{proof}

\subsection{Implementing the function approximation} \label{sec:finite_core:approx}

\begin{figure}[t]
	\centering
	\subfloat[True value]{ \label{fig:approximation_example:true}
		\begin{tikzpicture}
			\begin{axis}[xmin=0,xmax=53,ymax=1,samples=50,width=0.31\linewidth,height=3cm,
					axis x line = middle, axis y line* = middle,
					enlarge y limits=0,
					xlabel=\empty,ylabel=\empty,
					xtick={0,10,20,30,40,50},
					ytick={0,0.5,1},yticklabels={0,0.5,1}]
				\addplot[samples=50] plot coordinates {
					(0, 0)
					(18, 0)
					(20, 0.1)
					(21, 0.4)
					(22, 0.45)
					(25, 0.5)
					(40, 0.5)
					(41, 0.9)
					(45, 1)
					(53, 1)
				};
			\end{axis}
		\end{tikzpicture}
	}
	\subfloat[Simple approx.]{ \label{fig:approximation_example:simple}
		\begin{tikzpicture}
			\begin{axis}[xmin=0,xmax=53,ymax=1,samples=50,width=0.31\linewidth,height=3cm,
					axis x line = middle, axis y line* = middle,
					enlarge y limits=0,
					xlabel=\empty,ylabel=\empty,
					xtick={0,10,20,30,40,50},
					ytick={0,0.5,1},yticklabels={0,0.5,1}]
				\addplot[no marks,samples=50] plot coordinates {
					(0, 0)
					(18, 0)
					(20, 0.1)
					(21, 0.4)
					(22, 0.45)
					(25, 0.5)
					(40, 0.5)
					(41, 0.9)
					(45, 1)
					(53, 1)
				};
				\addplot[densely dashed,samples=50] plot coordinates {
					(0,  0)
					(10, 0)
					(10, 0.1)
					(20, 0.1)
					(20, 0.5)
					(30, 0.5)
					(40, 0.5)
					(40, 1)
					(53, 1)
				};
				\addplot[only marks,mark=*] plot coordinates {
					(0,  0)
					(10, 0)
					(20, 0.1)
					(30, 0.5)
					(40, 0.5)
					(50, 1)
				};
			\end{axis}
		\end{tikzpicture}
	}
	\subfloat[Adaptive approx.]{ \label{fig:approximation_example:adaptive}
		\begin{tikzpicture}
			\begin{axis}[xmin=0,xmax=53,ymax=1,samples=50,width=0.31\linewidth,height=3cm,
					axis x line = middle, axis y line* = middle,
					enlarge y limits=0,
					xlabel=\empty,ylabel=\empty,
					xtick={0,10,20,30,40,50},
					ytick={0,0.5,1},yticklabels={0,0.5,1}]
				\addplot[samples=50] plot coordinates {
					(0, 0)
					(18, 0)
					(20, 0.1)
					(21, 0.4)
					(22, 0.45)
					(25, 0.5)
					(40, 0.5)
					(41, 0.9)
					(45, 1)
					(53, 1)
				};
				\addplot[densely dashed,samples=50] plot coordinates {
					(18, 0)
					(18, 0.4)
					(21, 0.4)
					(21, 0.5)
					(40, 0.5)
					(40, 1)
					(41, 1)
					(52, 1)
				};
				\addplot[only marks,mark=*] plot coordinates {
					(18, 0)
					(21, 0.4)
					(40, 0.5)
					(45, 1)
				};
			\end{axis}
		\end{tikzpicture}
	}
	\caption{An example for different function approximation schemes which could be used to implement \textsc{UpdateBound} and \textsc{GetBound}.
		The graphs depict an arbitrarily chosen, monotonous function by a solid line and the corresponding approximation returned by the approximation scheme by a dashed line.
		From left to right, we have example bounds, which agree with the dense representation, followed by our sparse approach, which over-approximates the bounds, but requires less memory, and finally an adaptive approach, which closely resembles the precise bounds while consuming less memory.
		Dots represent the values stored by the sparse and adaptive approach.
		For readability, we do not depict the accumulating errors of the approximative methods in this figure.
	} \label{fig:approximation_example}
\end{figure}
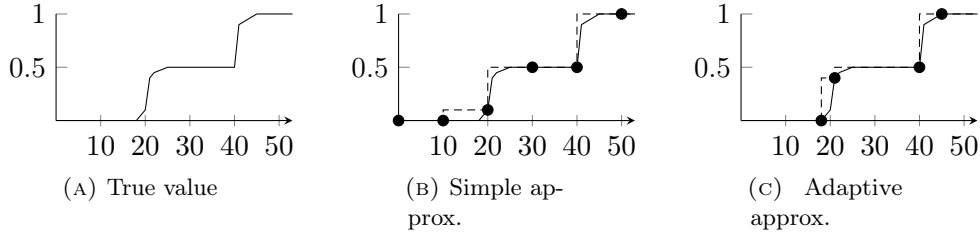

To illustrate the flexibility of our approach, we sketch several ideas for the implementation of \textsc{UpdateBound} and \textsc{GetBound} in Figure~\ref{fig:approximation_example}.
A concrete discussion and implementation of the more complex approach is left for future work.

The first, trivial implementation is dense storage, i.e.\ explicitly storing a table mapping $\States \times \{0, \dots, n - 1\} \to [0, 1]$.
This table representation consumes an unnecessary amount of memory, since we do not need exact values in order to just guide the exploration.
Hence, in our implementation, we use a simple sparse approach where we only store the value every $K$ steps, where $K$ is manually chosen.
Note that if we choose $K = 1$ we again obtain the \enquote{dense} approach.
This approach is depicted in Figure~\ref{fig:approximation_example:simple} for $K = 10$.
Every black dot represents a stored value, the dashed lines represent the value returned by \textsc{GetBound}.
We highlight that this approximation introduces accumulating errors.
This may actually prevent convergence, since these accumulated errors alone could be larger than $\varepsilon$ and thus the stopping criterion of the algorithm in Line~\ref{alg:learn_finite_core:loop} never is satisfied.
To overcome this issue, our implementation repeatedly computes the exact values for all explored states with a simple $n$-step value iteration, updating the sparsely stored value appropriately.
Note that such an update can be achieved with linear memory.

A more advanced idea is sketched in Figure~\ref{fig:approximation_example:adaptive}.
Using more sophisticated function approximation methods, we could adaptively choose which values are stored.
For example, there might be regions where the value of the function changes drastically and we should store more details there.
In the figure, this happens around 20 and 40 steps, respectively.

\subsection{Stability and its applications}

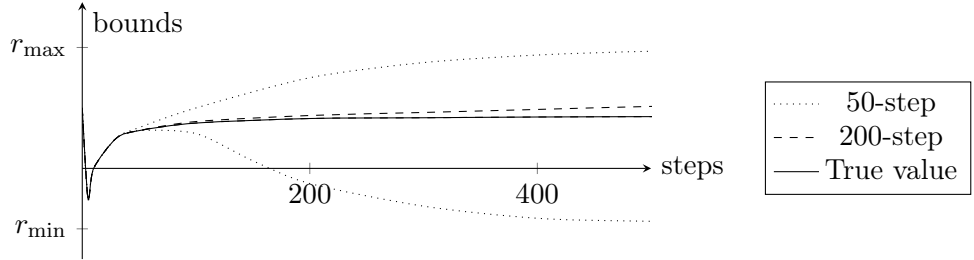
\begin{figure}[t]
	\centering
	\begin{tikzpicture}
		\begin{axis}[xmin=0,xmax=500,ymin=-6,ymax=11,samples=50,width=0.6\textwidth,height=5cm,
				axis lines = middle,
				enlarge y limits=0,
				legend style ={at={(axis cs:600,2)},anchor=west, draw=black, fill=white, align=left,
					nodes={transform shape}},
				x label style={at={(current axis.right of origin)},anchor=west},
				smooth,
				xlabel={steps},ylabel={bounds},
				xtick={0,200,400},
				ytick={-4,8},
				yticklabels={$r_{\min}$,$r_{\max}$},
				anchor=center]
			\addplot[dotted] coordinates {
				(0,4)
				(5,-2)
				(10,0)
				(30,2)
				(50,2.5)
				(100,2.3)
				(150,0.5)
				(200,-1)
				(300,-2.5)
				(400,-3.3)
				(500,-3.5)
				(800,-3.8)
				(900,-3.9)
			};
			\addplot[dotted,forget plot] coordinates {
				(0,4)
				(5,-2)
				(10,0)
				(30,2)
				(50,2.7)
				(100,4)
				(200,6)
				(300,7)
				(400,7.5)
				(500,7.75)
				(600,7.8)
				(700,7.9)
				(800,7.95)
				(900,7.995)
				(1000,8)
			};
			\addlegendentry{$50$-step};

			\addplot[dashed] coordinates {
				(0,4)
				(5,-2)
				(10,0)
				(30,2)
				(50,2.5)
				(100,3)
				(200,3.3)
				(300,3.35)
				(400,3.4)
				(500,3.425)
				(600,3.45)
				(700,3.475)
				(800,3.49)
				(900,3.5)
				(1000,3.5)
			};
			\addplot[dashed,forget plot] coordinates {
				(0,4)
				(5,-2)
				(10,0)
				(30,2)
				(50,2.5)
				(100,3.1)
				(200,3.5)
				(300,3.7)
				(400,3.9)
				(500,4.1)
				(600,4.2)
				(700,4.3)
				(800,4.4)
				(900,4.5)
				(1000,4.55)
			};
			\addlegendentry{$200$-step};

			\addplot[solid] coordinates {
				(0,4)
				(5,-2)
				(10,0)
				(30,2)
				(50,2.5)
				(100,3)
				(200,3.3)
				(300,3.35)
				(400,3.4)
				(500,3.425)
				(600,3.45)
				(700,3.475)
				(800,3.49)
				(900,3.5)
				(1000,3.5)
			};
			\addlegendentry{True value};
		\end{axis}
	\end{tikzpicture}
	\caption{
		A schematic plot for an average reward extrapolation analysis on step bounded cores.
		The solid line represents the true value, while the dotted and dashed lines are the respective upper and lower bounds computed for a $50$ and $200$-step core.
		Note that the second dashed line (lower bound on the $200$ core) coincides with the solid line (true value).
	} \label{fig:average_reward_schematic}
\end{figure}

In this section, we explain the idea of a core's \emph{stability}.
Given an $n$-step core $\FiniteCore$, we can easily compute the probability $\ProbabilityMax[\reach<N> \overline{\FiniteCore}]$ of exiting the core within $N > n$ steps using, e.g., value iteration.
The rate of increase of this exit probability intuitively gives us a measure of quality for a particular core.
Should it rapidly approach $1$ for increasing $N$, we know that the system's behaviour may change drastically within a few more steps.
If instead this probability remains small even for large $N$, we can compute properties with a large step bound on this core with tight guarantees.
We define stability as the whole function mapping the step bound $N$ to the exit probability, since this gives a more holistic view on the system's behaviour than a singular value.
This advantage becomes more apparent in the experimental evaluation.
In the following, we give an overview of how finite cores and the idea of stability can be used for analysis and interpretation, helping to design and understand complex systems.

As we have argued above, infinite-horizon properties may be deceiving, since in reality (unrecoverable) errors are bound to happen eventually.
Consequently, one might be interested in a \enquote{very large}-horizon analysis instead.
Unfortunately, such an analysis scales linearly both with the number of transitions and the horizon.
Considering that many systems have millions of states, an analysis with a horizon of only $10{,}000$ steps is far beyond the reach of existing tools.
We first explain how stable cores can be used for efficient extrapolation to such large horizons.

For simplicity, we consider reachability and later argue how to transfer this idea to other objectives.
We apply the ideas of interval iteration as used in, e.g., \cite{haddad2014reachability,atva14}, as follows.
Intuitively, since we have no knowledge of the partially explored states, we simply assume the worst and best case for them, i.e.\ assign a lower bound of $0$ and an upper bound of $1$.
Furthermore, any explored target state is assigned a lower and upper bound of $1$, as we know for sure that we reach our goal there.
By applying interval iteration we can obtain bounds on the $N$-step and even unbounded reachability.
Because of the core property, the bounds for $N \leq n$ necessarily are smaller than $\varepsilon$.
But, for larger $N$, i.e.\ $N > n$, there are no formal guarantees given by the core property---it might be the case that the core is left with probability $1$ in $n + 1$ steps.
Nevertheless, in practice this approach allows us to get good approximations even for much larger bounds.
We even observe that the computation of an $n$-step core and subsequent approximation of the desired property often is faster than directly computing the $N$-step property, as shown in our experimental evaluation.

For LTL and parity objectives, we can preprocess the obtained $n$-core by identifying the winning MECs and then applying the reachability idea.
This yields bounds on the probability of satisfying the given objective on the core.
In the case of mean-payoff, we again require lower and upper bounds on the rewards $r_{\min}$ and $r_{\max}$ of the system in order to properly initialize the unknown values.
Then, with the same approach, we can compute bounds on the $N$-step average reward by simply assigning the lower and upper bounds $r_{\min}$ and $r_{\max}$ to all unexplored states instead of $0$ and $1$.
See Figure~\ref{fig:average_reward_schematic} for a schematic plot of this analysis.
Here, the $50$-step core is too coarse for any reasonable analysis, it is unstable and can be exited with high probability.
On the other hand, the $200$-step core is very stable and accurately describes the system's behaviour for a longer period of time.
Noticeably, it also contains a MEC guaranteeing a lower bound on the average reward, hence the lower bound actually agrees with the true value.
Since the system may be significantly larger than the bounded cores or even infinitely large, this analysis potentially is much more efficient than analysis of the whole system.

Note that we cannot use this method to obtain arbitrarily precise results.
Given some $n$-step core and a particular (step bounded) property, there is a maximal precision we can achieve, depending on the property and the structure of the model.
Hence, this method primarily is useful to quickly obtain an overview of a system's behaviour instead of verifying a particular property.
As we have argued, one cannot avoid considering a particular part of the state space in order to obtain an $\varepsilon$-precise result.
Nevertheless, the smaller $n$-step core may provide valuable insights in a system, quickly giving a good overview of its behaviour or potential design flaws.
For example, an engineer could repeatedly run this analysis while formalizing or designing a system to quickly detect mistakes in the formalization or get a brief summarization of the systems performance.
Moreover, the fact that the system drastically changes its behaviour after $n$ steps may also provide valuable insights.

We recall that the presented algorithm can incrementally refine cores.
For example, if a $100$-step core does not yield a sufficiently precise extrapolation, the algorithm can reuse the computed core in order to construct a $200$-step core.
By applying this idea in an interactive loop, one can extract a condensed representation of the systems behaviour automatically, with the possibility for further refinements until the desired level of detail has been obtained.


\section{Experimental Evaluation} \label{sec:experiments}

In this section, we give practical results for our algorithms on several examples, both the hand-crafted plane model and models from case studies.
In the interest of space and readability, we hand-picked some noteworthy results from our overall evaluation.
Further details together with evaluation results on the complete PRISM benchmark suite \cite{DBLP:conf/qest/KwiatkowskaNP12}, can be found in Section~\ref{sec:appendix:evaluation}.

\subsection{Implementation Details}

We implemented our approach in Java, using PRISM 4.5 \cite{PRISM} as a library for parsing its modelling language and basic computations, verifying the correctness of our results.
Our implementation supports Markov chains, continuous-time Markov chains (CTMC, via embedding or uniformization \cite[Ch.~11.5]{puterman}) and Markov decision processes.
Further, we implemented our own version of some utility classes, e.g., explicit MDP representation and MEC decomposition. 
We point out that fine-tuning some parameters of the implementation (e.g., how often \textsc{UpdateECs} computes a full MEC decomposition) significantly impacts performance on some models.
This suggests that by investing additional effort into choosing these parameters heuristically the runtime could be improved further.

In \cite{atva14}, the authors presented several potential sampling heuristics, i.e.\ implementations of \textsc{GetPath}.
We evaluated some of the presented heuristics together with additional ones.
As reported in \cite{atva14}, it turns out that first selecting an action maximizing the expected upper bound and then selecting a successor in a weighted, guided fashion yields the best overall performance.
In particular, we sample a successor weighted by the respective upper bound, i.e.\ after selecting an action $a$ we randomly select a successor state $s'$ with probability proportional to $\upperbound(s') \cdot \transitions(s, a, s')$ or $\textsc{GetBound}(s', r) \cdot \transitions(s, a, s')$, respectively.
A detailed explanation and comparison between different sampling heuristics is presented in Section~\ref{sec:appendix:heuristics}.
In the following, we only consider the guided approach, since it consistently yielded the best performance.

\subsection{Models}

In our evaluation, we consider the following models.
\texttt{airplane} is our running example from Figure~\ref{fig:example_full} and Figure~\ref{fig:example_loop}, respectively.
All other models are taken from the PRISM benchmark suite \cite{DBLP:conf/qest/KwiatkowskaNP12}\footnote{Also available online at \url{http://www.prismmodelchecker.org/casestudies/}.}.
We briefly describe the models and how the associated parameters change them.

In \texttt{airplane}, the parameter \texttt{return} controls whether a return trip is possible and \texttt{size} quadratically influences the size of the \enquote{recovery} region.
\texttt{zeroconf} \cite{kwiatkowska2006performance} describes the IPv4 Zeroconf Protocol with \texttt{N} hosts, the number \texttt{K} of probes to send, and a probability \texttt{loss} of losing a message.
\texttt{wlan} \cite{kwiatkowska2002probabilistic} is a model of two WLAN stations in a fixed network topology sending messages on the shared medium, potentially leading to collisions.
\texttt{brp} is a DTMC modelling a file transfer of \texttt{N} chunks with bounded number \texttt{MAX} of retries per chunk.
Finally, \texttt{cyclin} is a CTMC modelling the cell cycle control in eukaryotes with \texttt{N} molecules.
We analyse this model using uniformization, converting it to a DTMC.

\subsection{Results}

We evaluated our implementation on an Intel Xeon \texttt{E5-2630} 2.20 GHz CPU, allocating one core (using \texttt{taskset}) and 8 GB of RAM to the Java process (using \texttt{-Xmx8G}).
We used a default precision of $10^{-6}$ and a timeout of 15 minutes for all experiments.
The evaluation is performed with the help of GNU \texttt{parallel} \cite{Tange2011a}.
The results for the infinite and finite construction are summarized in Table~\ref{tbl:results}.
We discuss them in the following sections.
Note that the results may vary due to the involved randomization.

\begin{table}[t]
	\caption{
		Summary of our experimental results on several models and configurations for the unbounded and step-bounded core learning.
		The \enquote{PRISM} column shows the total number of states and construction time when explored with the explicit engine.
		The following columns show the size and total construction time of a $10^{-6}$-core and a $100$-step $10^{-6}$-core, respectively.
		While building the step-bounded core for \texttt{brp}, we used the \enquote{dense} storage approach, since the simple approximation yielded unreliable performance.
	} \label{tbl:results}
	\centering
	\begin{tabular}{llrrrrrr}
		Model                                   & Param.               & \multicolumn{2}{c}{PRISM}  &  \multicolumn{2}{c}{Core}   & \multicolumn{2}{c}{$100$-Core} \\
		\toprule
		\texttt{zeroconf}                       & $100; 5; 0.1$        &        496{,}291 &    13 s &       820 &             1 s &  1{,}087 &                 1 s \\
		(\texttt{N}; \texttt{K}; \texttt{loss}) & $100; 10; 0.1$       & $3.0 \cdot 10^6$ &    77 s &       706 &             1 s &  1{,}036 &                 1 s \\
		                                        & $100; 15; 0.1$       & $4.7 \cdot 10^6$ &   120 s &       766 &             1 s &  1{,}192 &                 1 s \\
		\midrule
		\texttt{airplane}                       & $100; \texttt{ff}$   &         10{,}208 &     1 s &         6 &             0 s &        6 &                 0 s \\
		(\texttt{size}; \texttt{return})        & $100; \texttt{tt}$   &         20{,}413 &     1 s & \multicolumn{2}{c}{TIMEOUT} &       33 &                 0 s \\
		                                        & $10000; \texttt{ff}$ & \multicolumn{2}{c}{MEMOUT} &         6 &             0 s &        6 &                 0 s \\
		                                        & $10000; \texttt{tt}$ & \multicolumn{2}{c}{MEMOUT} & \multicolumn{2}{c}{TIMEOUT} &       32 &                 0 s \\
		\midrule
		\texttt{brp}                            & $20; 10$             &          2{,}933 &     0 s &   1{,}437 &             1 s &  1{,}359 &                 0 s \\
		(\texttt{N}; \texttt{MAX})              & $20; 100$            &         26{,}423 &     1 s &   1{,}442 &             1 s &  1{,}324 &                 0 s \\
		                                        & $20; 1000$           &        261{,}323 &     8 s &   1{,}437 &             1 s &  1{,}336 &                 0 s \\
		\midrule
		\texttt{wlan}                           & ---                  &        345{,}000 &   8.2 s & 345{,}000 &            74 s & 35{,}998 &                46 s \\
		\midrule
		\texttt{cyclin}                         & $4$                  &        431{,}101 &    18 s & \multicolumn{2}{c}{TIMEOUT} & 11{,}465 &                 4 s \\
		(\texttt{N})                            & $5$                  & $2.3 \cdot 10^6$ &   118 s & \multicolumn{2}{c}{TIMEOUT} & 36{,}613 &                13 s \\
		\bottomrule
	\end{tabular}
\end{table}

\subsubsection{Infinite Cores}
As already explained in \cite{atva14}, the \texttt{zeroconf} model is very well suited for this type of analysis, since a lot of the state space is hardly reachable.
In particular, most states are a result of several consecutive collisions and message losses, which is very unlikely.
Consequently, a very small part of the model already satisfies the core property.
The size of the core remains practically constant when increasing the parameter $\texttt{K}$, as only unimportant states are added to the system.
We note that the order of magnitude of explored states is very similar to the experiments from \cite{atva14}.
Similarly, in the \texttt{airplane} model a significant number of states is dedicated to recovering from an unlikely error.
Hence, a small core exists independent of the total size of the model.
The \texttt{brp} model shows applicability of the approach to Markov chains.
In line with the other results, when scaling up the maximal number of allowed errors, the size of the core changes sub-linearly, since repeated errors are increasingly unlikely.

The \texttt{airplane} model with return trip (and \texttt{cyclin} to a lesser extent) actually shows a structural weakness of our purely sampling-/VI-based approach:
Recall that the \enquote{non-recovery} region, i.e.\ the round-trip path before an error occurs, is not a MEC, however the probability of exiting is very low, namely $\tau \ll \varepsilon$ per round-trip.
This leads to two problems.
Firstly, any sampling based approach which is influenced by the transition probabilities (including our weighted approach) only rarely explores the eventually important recovery region.
Secondly, even if a path is sampled in that region the update-computation only propagates a miniscule fraction of the obtained information back to the round trip states.
Here, a hybrid approach combined with strategy iteration might be useful.
We emphasize that this not an inherent issue of the \enquote{core} idea, but rather an inherent issue of value iteration---computing a reachability property on this model using value iteration takes very long due to the latter reason, too.

\textbf{Comparison to \cite{atva14}:} We also executed the tool presented in \cite{atva14} where applicable (only MDP are supported).
We tested the tool both with an unsatisfiable property ($\reach \emptyset$), i.e.\ approximating the probability of reaching the empty set, which corresponds to constructing a core, and an actual property.
We used the \texttt{MAX\_DIFF} heuristic of \cite{atva14}, since it was suggested to be the best-performing setting.
Especially on the $\reach \emptyset$ property, our tool consistently outperformed the previous one in terms of time and memory by up to several orders of magnitude.
We suspect that this is mostly due to a more efficient implementation, especially since the number of explored states was similar.

\subsubsection{Finite Cores}

As expected, the finite core construction yields good results on the \texttt{airplane} model, constructing only a small fraction of the state space.
On the real-world models \texttt{wlan} and \texttt{cyclin}, the constructed $100$-step core is significantly smaller than the whole model.
For \texttt{wlan}, the construction of the respective cores unfortunately takes longer than building the whole model when using the \enquote{simple} approximation.
Nevertheless, model checking on the explored sub-system supposedly terminates significantly faster since only a much smaller state space is investigated, and the core can be re-used for more queries.

During our experiments, we used the \enquote{simple} approximation approach ($K = 5$) introduced in Section~\ref{sec:finite_core:approx}.
Interestingly, this approach actually yielded significant speed-ups and a smaller core on the \texttt{cyclin} model compared to using the \enquote{dense} approximation.
On the other hand, \enquote{dense} terminated much faster with a comparable core size on both \texttt{wlan} and \texttt{brp}, with a speed-up of nearly an order of magnitude.
We conjecture that this difference potentially is related to \texttt{cyclin} being a uniformized CTMC, resulting in a particular structure.

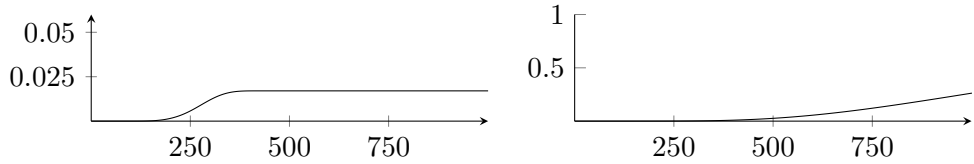
\begin{figure}[t]
	\centering
	\begin{tikzpicture}
		\begin{axis}[xmin=0,xmax=1001,ymax=0.06,samples=50,width=0.45\textwidth,height=3cm,
				axis lines = middle,
				enlarge y limits=0,
				scaled y ticks=false, yticklabel style={
					/pgf/number format/fixed,
					/pgf/number format/precision=3
				},
				xlabel=\empty,ylabel=\empty,
				xtick={250,500,750},
				ytick={0,0.025,0.05}]
			\addplot[no marks,black,-] table [x expr=\coordindex, y index=0] {data/stability.wlan.txt};
		\end{axis}
	\end{tikzpicture}
	\begin{tikzpicture}
		\begin{axis}[xmin=0,xmax=1000,ymax=1,samples=50,width=0.45\textwidth,height=3cm,
				axis x line = middle, axis y line* = middle,
				enlarge y limits=0,
				xlabel=\empty,ylabel=\empty,
				xtick={250,500,750},
				ytick={0,0.5,1}]
			\addplot[no marks,black,-] table [x expr=\coordindex, y index=0] {data/stability.cyclin.txt};
		\end{axis}
	\end{tikzpicture}
	\caption{Stability analysis of the \texttt{wlan} (left) and \texttt{cyclin}($\texttt{N}=4$) (right) 100-step core.
		The graphs show the probability of exiting the respective core within the given amount of steps.
		The $y$ axis of the \texttt{wlan} graph is scaled.
	} \label{fig:results_stability}
\end{figure}

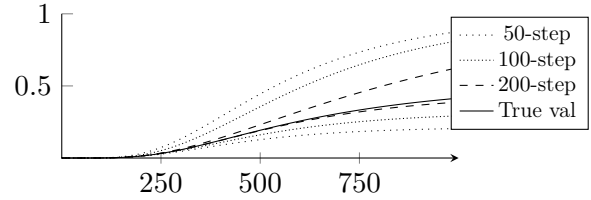
\begin{figure}[t]
	\centering
	\setlength{\tabcolsep}{5pt}
	\begin{tabular}{rrrr}
		\multirow{2}{4em}{Method} & \multicolumn{2}{c}{Time} & \multirow{2}{3em}{States} \\
		                          & Model &           Bounds &                           \\
		\hline
		                 50 steps &  0.6s &             1.2s &                   3{,}006 \\
		                100 steps &  0.8s &             1.2s &                   3{,}590 \\
		                200 steps &  1.5s &             1.5s &                   7{,}769 \\
		\hline
		                 Complete &   13s &              49s &                 431{,}101
	\end{tabular}
	\begin{tikzpicture}[baseline]
		\begin{axis}[xmin=0,xmax=1000,ymin=0.0,ymax=1.0,samples=50,width=0.45\textwidth,height=3.5cm,
				axis x line = middle, axis y line* = middle,
				enlarge y limits=0,
				legend style ={anchor=north west, draw=black, fill=white, align=left,
					nodes={scale=0.75, transform shape}},
				legend image post style={scale=0.75},
				smooth,
				xlabel=\empty,ylabel=\empty,
				xtick={0,250,500,750},
				ytick={0,0.5,1},
				anchor=center]
			\addplot[dotted] table [x expr=\coordindex, y index=0] {data/extra.50.lower.txt};
			\addplot[dotted,forget plot] table [x expr=\coordindex, y index=0] {data/extra.50.upper.txt};
			\addlegendentry{$50$-step};

			\addplot[densely dotted] table [x expr=\coordindex, y index=0] {data/extra.100.lower.txt};
			\addplot[densely dotted,forget plot] table [x expr=\coordindex, y index=0] {data/extra.100.upper.txt};
			\addlegendentry{$100$-step};

			\addplot[dashed] table [x expr=\coordindex, y index=0] {data/extra.200.lower.txt};
			\addplot[dashed,forget plot] table [x expr=\coordindex, y index=0] {data/extra.200.upper.txt};
			\addlegendentry{$200$-step};

			\addplot[solid] table [x expr=\coordindex, y index=0] {data/extra.complete.txt};
			\addlegendentry{True val};
		\end{axis}
	\end{tikzpicture}
	\caption{
		Overview of an extrapolation analysis for \texttt{cyclin}($N = 4$).
		We computed several step-bounded cores with precision $10^{-3}$.
		On these, we computed bounds of a reachability query with increasing step bound.
		The table on the left lists the time for model construction + computation of the bounds for 1000 steps and the size of the constructed model.
		The plot on the right shows the upper and lower bounds computed for each core together with the true value.
		Observe that for growing step-size of the core, the approximation naturally gets more precise.
	} \label{fig:results_extrapolation}
\end{figure}

Finally, we applied the idea of stability from Section~\ref{sec:finite_cores:stability} on the \texttt{wlan} and \texttt{cyclin} models, with results outlined in Figure~\ref{fig:results_stability}.
Interestingly, for the \texttt{wlan} model, the escape probability stabilizes at roughly $0.017$ and we obtain the exact same probability for \emph{all} heuristics (see Section~\ref{sec:appendix:heuristics}), even for $N = 10{,}000$.
This suggests that by building the $100$-step core we identified a very stable sub-system of the whole model.
Recall that the \texttt{wlan} model has roughly $3.5 \cdot 10^5$ states in total, while the identified subsystem comprises only $10\%$ of these states.
This means that most of the long term behaviour is described by only a fraction of the states.
Additionally, we observe that the crucial actions leading to the other $90\%$ of the state space happen at around 200--400 steps and the system is stable afterwards.
This information is only visible since we considered the stability function as a whole instead of a single number.

For the \texttt{cyclin} model, we instead observe a continuous rise of the exit probability.
Nevertheless, even with 500 additional steps, the core still is only exited with a probability of roughly $5.5\%$ and thus closely describes the system's behaviour.
On the \texttt{cyclin} model, we also applied our idea of extrapolation.
The results are summarized in Figure~\ref{fig:results_extrapolation}.
To show how performant this approach is, we reduced the precision of the core computation to $10^{-3}$.
Despite this coarse accuracy, we are able to compute accurate bounds on a $1000$-step reachability query\footnote{We used the (arbitrarily chosen) query \texttt{dim > CYCLIN / 4}.} over 10 times faster by only building the $200$-step core instead of constructing the full model.
In particular, we obtain the result significantly faster than the construction of the whole model.
These results suggest that our idea of using the cores for extrapolation in order to quickly gain understanding of a model has a vast potential.

\subsubsection{PRISM benchmark suite}

For readability, we briefly summarize our findings from Section~\ref{sec:appendix:results} here.
Firstly, we observe that our methods are sensitive to the structure of the model and the particular guidance heuristic used.
Weighted guidance, incorporating both computed bounds and transition probabilities, shows the best potential out of the investigated heuristics.
Moreover, the results suggest that small cores are not too common in the models of the PRISM benchmark suite.
However, we conjecture that this is mostly due to the specific structure of these models.
In particular, most of them describe abstract protocols, where probabilistic branching is only used in a few critical locations and is of a special structure; a significant part of the system size stems from non-deterministic interleaving of parallel processes.
For example, in the \texttt{firewire} protocol, randomness is only used to select either a \texttt{fast} or \texttt{slow} mode to eventually resolve ties.
This often results in rather large probabilities and thus hardly any \enquote{unimportant} states.
For such models, our method naturally is not applicable.
However, given real-world constraints, many low-probability events are introduced to the model, e.g., hardware failures, sensor noise, or transmission errors due to environmental influences.
These low probability errors allow for non-trivial cores, as is the case with, for example, the \texttt{zeroconf} or \texttt{brp} model.

\section{Conclusion}

We have presented a new framework for approximate verification of probabilistic systems via partial exploration and applied it to both Markov chains and Markov decision processes.
Our evaluation shows that, depending on the structure of the model, this approach can yield significant state space savings and thus reduction in model checking times.
Our central idea---finding relevant sub-parts of the state space---can easily be extended to further models, e.g., stochastic games, and objectives, e.g., mean payoff.
We have also shown how this idea can be transferred to the step-bounded setting and derived the notion of stability.
This in turn allows for an efficient analysis of long-run properties and strongly connected systems.

Future work includes implementing a more sophisticated function approximation for the step-bounded case, e.g., as depicted in Figure~\ref{fig:approximation_example:adaptive}.
Note that this adaptive method could yield further insight in the model by deriving points of interest, i.e.\ an interval of remaining steps where the exit probability significantly changes.
These breakpoints might indicate a significant change in the systems behaviour, e.g., the probability of some error occurring not being negligible any more, yielding interesting insights into the structure of a particular model.
For example, in the bounds of Figure~\ref{fig:approximation_example}, the regions around 20 and 40 steps, respectively, seems to be of significance.

Moreover, a sophisticated sampling heuristic to be used in \textsc{SamplePath} could be of interest.
For example, one could apply an advanced machine learning technique here, which also considers state labels or previous decisions and their outcomes.
In terms of performance, one could consider parallelizing the sampling procedures and applying rare-event detection mechanisms to reduce the number of samples needed to find a core.
Another point for performance improvements is the use of faster MEC detection algorithms.
In particular, \cite{CH11} presents an incremental MEC decomposition algorithm which is able to maintain the set of MECs of a dynamically changing MDP.
Currently, our implementation recomputes the ECs from scratch every time.

In the spirit of \cite{atva14}, our approach also could be extended to a PAC algorithm for black-box systems.
Extensions to continuous time systems are also possible.
Practically, application of our methods to models derived from biological systems or chemical reactions could yield more satisfactory results, since such systems are much more \enquote{probabilistic} than the abstract protocols considered in our evaluation.

Further interesting variations are cores for discounted objectives \cite{DBLP:conf/soda/SidfordWWY18} or \enquote{cost-bounded} cores, a set of states which is left with probability smaller than $\varepsilon$ given that at most $k$ cost is incurred.
This generalizes both the infinite (all edges have cost $0$) and the step bounded cores (all edges have cost $1$) and allows for a much wider range of analysis.


\bibliographystyle{alpha}
\bibliography{main}
\newpage

\appendix
\section{Evaluation Details} \label{sec:appendix:evaluation}

In this section, we provide further evaluation data and implementation details.

\subsection{Sampling heuristics} \label{sec:appendix:heuristics}
All of our implementations of \textsc{GetPath} sample paths originating from the initial state.
We consider two classes of sampling heuristics namely \emph{action-based} and \emph{graph-based}.

For action-based heuristics, we first select an action maximizing the expected upper bound, i.e.\ randomly choose an action from $\arg\max_{a \in \stateactions(s)} \sum_{s' \in \States} \transitions(s, a, s') \cdot \upperbound(s')$.
Once such an action $a$ is selected, we obtain a successor based on a scoring function $f(s, a, s')$, i.e.\ a state $s'$ is selected with probability proportional to $f(s, a, s')$.
The graph-based approach however essentially ignores the available actions and chooses successors directly based on a scoring function $f(s, s')$.

We considered three action-based heuristics, namely
\begin{itemize}
	\item \texttt{PROB} (\texttt{P}): $f(s, a, s') = \transitions(s, a, s')$,
	\item \texttt{WEIGHTED} (\texttt{W}): $f(s, a, s') = \upperbound(s') \cdot \transitions(s, a, s')$, and
	\item \texttt{DIFFERENCE} (\texttt{D}): $f(s, a, s') = \upperbound(s')$.
\end{itemize}
Moreover, we considered the following graph-based heuristics
\begin{itemize}
	\item \texttt{GRAPH-WEIGHTED} (\texttt{GW}): $f(s, s') = \upperbound(s') \cdot \max_{a \in \stateactions(s)} \transitions(s, a, s')$, and
	\item \texttt{GRAPH-DIFFERENCE} (\texttt{GD}): $f(s, s') = \upperbound(s')$.
\end{itemize}
In the finite-horizon case, the heuristics use $\textsc{GetBound}(s', r)$ instead of $\upperbound(s')$.
The \texttt{DIFFERENCE} heuristics corresponds to the \texttt{M-D} mode of \cite{atva14}, originally proposed in \cite{DBLP:conf/icml/McMahanLG05}.

\subsection{Setup}
We evaluated our methods on the complete PRISM benchmark set where applicable, i.e.\ all DTMC, MDP, and CTMC.
We treated all CTMCs via unifomization, determining the uniformization constant using PRISM's \texttt{getDefaultUniformisationRate()} method (which equals $1.02$ times the maximum exit rate).
For each considered model / parameter combination (215 in total), we (i)~built the complete model using PRISM's methods (ii)~built an unbounded core and (iii)~built step bounded cores for $n \in \{10, 100, 200, 500\}$, using each of our 5 heuristics, resulting in 26 executions per model and 5590 runs in total.
All cores are built with a precision requirement of $\varepsilon = 10^{-6}$.
Each execution was run with a time limit of 15 minutes and memory limit of 8 GB, bound to a single CPU core to avoid influences of potential parallelism.
Since a significant part of these models are quite large, exceeding the imposed limits by a huge margin, we encountered several time- and mem-outs.
For each evaluation, we thus describe in detail how we treated such failures when aggregating the values.

We note that the models of the PRISM benchmark suite are not too well suited for our methods overall.
We explain this issue in the following section.
Nevertheless, we evaluated our methods on this dataset, since it is one of the largest established datasets for this purpose.

\subsection{Evaluation Results} \label{sec:appendix:results}

In the first evaluation, we compare the overall performance of our methods and compare the different heuristics.
To this end, we compute for each method the average number of states, fraction of states compared to the original model, the number of failures, and the number of times the method succeeded where the complete model construction did not succeed.
The results are presented in Table~\ref{tbl:results_prism_all}.
Since a large number of these models comprise a single connected component (SCC or MEC), we separately consider all models where (i)~the complete model construction finished and (ii)~several components are found (or one not too large one), of which there are 117 in total.
The results are shown in Table~\ref{tbl:results_prism_many}.

\begin{table}
	\caption{
		Comparison of our heuristics on the complete PRISM benchmark suite (215 instances in total).
		The \enquote{Time} and \enquote{States} column are (arithmetic) averages over all instances solved by the particular method.
		\enquote{Fraction} describes the average fraction of states compared to the whole model, i.e.\ $\cardinality{\Core} / \cardinality{\States}$, considering only those instances where both the particular heuristic and the complete construction succeeded.
		\enquote{Failures} and \enquote{Success} are the percentage of models where the complete construction succeeded but the heuristic failed and vice versa.
		For example, the third line (\enquote{Unbounded \texttt{W}}) means that the \texttt{W} heuristic for unbounded cores took an average of 32 seconds to converge, resulting in 342306 states, with only 75\% of the states of the complete model built.
		Finally, on 42\% of the models this method failed but the complete construction finished in time, dually on 7\% of the models the heuristic succeeded but the complete construction did not.
	} \label{tbl:results_prism_all}
	\begin{tabular}{lrrrrr}
		                         &  Time & States & Fraction & Failures & Success \\
		\toprule
		Complete                 &   7 s & 268880 &      --- &    21 \% &     --- \\
		\midrule
		\multicolumn{6}{l}{Unbounded}                                             \\
		\hspace{1em} \texttt{W}  &  32 s & 342306 &    75 \% &    42 \% &    7 \% \\
		\hspace{1em} \texttt{P}  &  44 s & 149538 &    73 \% &    61 \% &    0 \% \\
		\hspace{1em} \texttt{D}  &  28 s & 344932 &    81 \% &    39 \% &    6 \% \\
		\hspace{1em} \texttt{GW} &  34 s & 339944 &    75 \% &    41 \% &    7 \% \\
		\hspace{1em} \texttt{GD} &  33 s & 362027 &    81 \% &    42 \% &    6 \% \\
		\midrule
		\multicolumn{6}{l}{Bounded 10}                                            \\
		\hspace{1em} \texttt{W}  &   5 s &   5579 &    15 \% &     3 \% &   21 \% \\
		\hspace{1em} \texttt{P}  &  11 s &    349 &    13 \% &    12 \% &   19 \% \\
		\hspace{1em} \texttt{D}  &   4 s &  14738 &    18 \% &     5 \% &   21 \% \\
		\hspace{1em} \texttt{GW} &   6 s &   5763 &    15 \% &     3 \% &   21 \% \\
		\hspace{1em} \texttt{GD} &   5 s &  14741 &    18 \% &     5 \% &   21 \% \\
		\multicolumn{6}{l}{Bounded 100}                                           \\
		\hspace{1em} \texttt{W}  &  39 s &  29387 &    52 \% &    29 \% &    6 \% \\
		\hspace{1em} \texttt{P}  & 127 s &   5047 &    50 \% &    49 \% &    5 \% \\
		\hspace{1em} \texttt{D}  &  31 s &  24935 &    65 \% &    33 \% &    5 \% \\
		\hspace{1em} \texttt{GW} &  34 s &  28812 &    51 \% &    29 \% &    6 \% \\
		\hspace{1em} \texttt{GD} &  22 s &  24405 &    64 \% &    33 \% &    5 \% \\
		\multicolumn{6}{l}{Bounded 200}                                           \\
		\hspace{1em} \texttt{W}  &  41 s &  22415 &    61 \% &    35 \% &    3 \% \\
		\hspace{1em} \texttt{P}  & 127 s &   7424 &    64 \% &    60 \% &    1 \% \\
		\hspace{1em} \texttt{D}  &  35 s &  24796 &    74 \% &    45 \% &    1 \% \\
		\hspace{1em} \texttt{GW} &  40 s &  22665 &    60 \% &    33 \% &    3 \% \\
		\hspace{1em} \texttt{GD} &  32 s &  23981 &    73 \% &    44 \% &    1 \% \\
		\multicolumn{6}{l}{Bounded 500}                                           \\
		\hspace{1em} \texttt{W}  &  37 s &  14926 &    66 \% &    38 \% &    3 \% \\
		\hspace{1em} \texttt{P}  &  96 s &   9334 &    74 \% &    67 \% &    0 \% \\
		\hspace{1em} \texttt{D}  &  39 s &  14723 &    77 \% &    46 \% &    1 \% \\
		\hspace{1em} \texttt{GW} &  37 s &  14084 &    65 \% &    37 \% &    3 \% \\
		\hspace{1em} \texttt{GD} &  34 s &  13983 &    76 \% &    44 \% &    1 \% \\
		\bottomrule
	\end{tabular}
\end{table}

\begin{table}
	\caption{
		Comparison of our heuristic on models of the PRISM benchmark suite which are not strongly connected (117 in total).
		We use the same metrics as in Table~\ref{tbl:results_prism_all}.
		We determine whether a model is strongly connected based on the results of the complete construction, hence only models where this construction succeeds are included.
		Consequently, the \enquote{Success} column is not relevant.
	} \label{tbl:results_prism_many}
	\begin{tabular}{lrrrr}
		                         & Time & States & Fraction & Failures \\
		\toprule
		Complete                 &  4 s & 217159 &      --- &      --- \\
		\midrule
		\multicolumn{5}{l}{Unbounded}                                  \\
		\hspace{1em} \texttt{W}  & 14 s & 121213 &    74 \% &     9 \% \\
		\hspace{1em} \texttt{P}  & 40 s &  72114 &    72 \% &    32 \% \\
		\hspace{1em} \texttt{D}  & 12 s & 122865 &    79 \% &     9 \% \\
		\hspace{1em} \texttt{GW} & 16 s & 121752 &    74 \% &    11 \% \\
		\hspace{1em} \texttt{GD} & 15 s & 123926 &    79 \% &    11 \% \\
		\midrule
		\multicolumn{5}{l}{Bounded 10}                                 \\
		\hspace{1em} \texttt{W}  &  0 s &    273 &    11 \% &     6 \% \\
		\hspace{1em} \texttt{P}  &  2 s &    244 &    10 \% &     6 \% \\
		\hspace{1em} \texttt{D}  &  0 s &    395 &    13 \% &     6 \% \\
		\hspace{1em} \texttt{GW} &  0 s &    273 &    11 \% &     6 \% \\
		\hspace{1em} \texttt{GD} &  0 s &    400 &    13 \% &     6 \% \\
		\multicolumn{5}{l}{Bounded 100}                                \\
		\hspace{1em} \texttt{W}  & 33 s &  20343 &    54 \% &    24 \% \\
		\hspace{1em} \texttt{P}  & 84 s &   7388 &    55 \% &    40 \% \\
		\hspace{1em} \texttt{D}  & 25 s &  19975 &    61 \% &    19 \% \\
		\hspace{1em} \texttt{GW} & 27 s &  19479 &    53 \% &    23 \% \\
		\hspace{1em} \texttt{GD} & 15 s &  19228 &    59 \% &    19 \% \\
		\multicolumn{5}{l}{Bounded 200}                                \\
		\hspace{1em} \texttt{W}  & 46 s &  25642 &    65 \% &    25 \% \\
		\hspace{1em} \texttt{P}  & 96 s &  10169 &    69 \% &    48 \% \\
		\hspace{1em} \texttt{D}  & 38 s &  29370 &    73 \% &    34 \% \\
		\hspace{1em} \texttt{GW} & 46 s &  25870 &    63 \% &    21 \% \\
		\hspace{1em} \texttt{GD} & 39 s &  29614 &    72 \% &    31 \% \\
		\multicolumn{5}{l}{Bounded 500}                                \\
		\hspace{1em} \texttt{W}  & 36 s &  15823 &    69 \% &    28 \% \\
		\hspace{1em} \texttt{P}  & 92 s &  10739 &    71 \% &    48 \% \\
		\hspace{1em} \texttt{D}  & 42 s &  16431 &    77 \% &    34 \% \\
		\hspace{1em} \texttt{GW} & 42 s &  15888 &    69 \% &    25 \% \\
		\hspace{1em} \texttt{GD} & 36 s &  15368 &    76 \% &    31 \% \\
		\bottomrule
	\end{tabular}
\end{table}

Overall, we see that the unguided, random sampling heuristic \textsf{P} often is significantly outperformed by the guided approaches \textsf{W} and \textsf{D} in terms of runtime.
Note that both the \enquote{Time} and \enquote{States} column only shows averages over all instances where the particular method succeeded, hence the averages size of \textsf{P}-cores seem lower than the others'.
When comparing the instances where all heuristics succeed, we see that the average sizes are very similar.
We also observe that the size and construction time (and failure rates) of step-bounded does not significantly increase when going from $200$ to $500$ steps.
This suggests that most of the considered models where the construction succeeds are quite \enquote{stable} after $200$ steps---otherwise, either the average size of the computed $500$-step cores or the number of timeouts would be much larger.
Finally, we see that the size of the identified cores are mostly independent of the used method (accounting for the difference in the averages due to timeouts).
When repeating the experiments, we further learned that the size of the identified cores usually deviate only by a few percent.
This suggests that the performance of our approach is quite stable, despite the large amount of involved randomization.

\begin{table}
	\caption{Comparison of our heuristics on each of the model groups of the PRISM benchmark suite.
		In the first four columns, we report the number of instances, the average number of states (of the models where the complete construction succeeded), the absolute number of failures of the complete construction and the average number of components of the model instances (MECs or SCCs, respectively).
		For each heuristic, we then report the average fraction of states, compared to the complete construction (the previous tables' \enquote{Fraction} column).
		We omit the \% sign to save space.
		For simplicity, whenever a heuristic failed, we simply assumed a ratio of $100 \%$, independent of the actual outcome of the complete construction.
		Further, we report the number of times where the heuristic approach succeeded over the complete construction (the previous tables' \enquote{Success} column).
		For readability, we sorted the models by the average state-reduction achieved by our methods.
	} \label{tbl:results_prism_by_model}
	\begin{tabular}{lrrrrrrrrrrrrrr}
		Model         & \# & States &  Comp. & Fail & \multicolumn{2}{c}{\texttt{W}} & \multicolumn{2}{c}{\texttt{P}} & \multicolumn{2}{c}{\texttt{D}} & \multicolumn{2}{c}{\texttt{GW}} & \multicolumn{2}{c}{\texttt{GD}} \\
		\toprule
		zeroconf      & 16 & 383919 &   5476 &    0 &  12 &                        0 &  11 &                        0 &  15 &                        0 &  13 &                         0 &  15 &                         0 \\
		zeroconf\_dl  & 10 &  91996 &   3820 &    0 &  23 &                        0 &  19 &                        0 &  29 &                        0 &  26 &                         0 &  34 &                         0 \\
		embedded      &  7 &   6013 &   4529 &    0 &  52 &                        0 & 100 &                        0 & 100 &                        0 &  52 &                         0 & 100 &                         0 \\
		wlan          &  7 & 969161 &      1 &    0 &  77 &                        0 & 100 &                        0 &  77 &                        0 &  77 &                         0 &  77 &                         0 \\
		nand          & 10 & 193247 & 193247 &    6 &  75 &                        6 & 100 &                        0 &  96 &                        4 &  75 &                         6 &  96 &                         4 \\
		brp           & 12 &   2302 &   2302 &    0 &  97 &                        0 &  96 &                        0 & 100 &                        0 &  97 &                         0 & 100 &                         0 \\
		\midrule
		wlan\_dl      &  7 & 596250 &  11671 &    4 & 100 &                        4 & 100 &                        0 & 100 &                        4 & 100 &                         4 & 100 &                         4 \\
		crowds        & 16 &  87287 &  56140 &    4 & 100 &                        3 & 100 &                        0 & 100 &                        3 & 100 &                         3 & 100 &                         3 \\
		coin          &  6 &  11616 &     27 &    0 & 100 &                        0 & 100 &                        0 & 100 &                        0 & 100 &                         0 & 100 &                         0 \\
		csma          &  9 & 389136 &     11 &    3 & 100 &                        0 & 100 &                        0 & 100 &                        0 & 100 &                         0 & 100 &                         0 \\
		egl           & 16 &  95230 &  95230 &   12 & 100 &                        0 & 100 &                        0 & 100 &                        0 & 100 &                         0 & 100 &                         0 \\
		firewire      & 20 & 342520 &   1214 &    4 & 100 &                        1 & 100 &                        1 & 100 &                        1 & 100 &                         1 & 100 &                         1 \\
		herman        &  7 &   6241 &      5 &    0 & 100 &                        0 & 100 &                        0 & 100 &                        0 & 100 &                         0 & 100 &                         0 \\
		leader\_sync  &  9 &    758 &    633 &    0 & 100 &                        0 & 100 &                        0 & 100 &                        0 & 100 &                         0 & 100 &                         0 \\
		\midrule
		cluster       &  9 & 396800 &      1 &    1 & 100 &                        0 & 100 &                        0 & 100 &                        0 & 100 &                         0 & 100 &                         0 \\
		fms           & 10 & 339031 &      1 &    3 & 100 &                        0 & 100 &                        0 & 100 &                        0 & 100 &                         0 & 100 &                         0 \\
		kanban        &  7 & 612813 &      1 &    2 & 100 &                        0 & 100 &                        0 & 100 &                        0 & 100 &                         0 & 100 &                         0 \\
		mapk\_cascade &  8 & 316918 &      1 &    2 & 100 &                        0 & 100 &                        0 & 100 &                        0 & 100 &                         0 & 100 &                         0 \\
		poll          & 18 & 419430 &      1 &    3 & 100 &                        0 & 100 &                        0 & 100 &                        0 & 100 &                         0 & 100 &                         0 \\
		tandem        & 11 & 310465 &      1 &    2 & 100 &                        0 & 100 &                        0 & 100 &                        0 & 100 &                         0 & 100 &                         0 \\
		\bottomrule
	\end{tabular}
\end{table}

Since the benchmark set contains many vastly different models, both in terms of structure and size, we further report results grouped by the different model types for the unbounded core construction in Table~\ref{tbl:results_prism_by_model}.
As already expected, we see that our methods are highly dependent on the model structure.
Moreover, we see that on all models except \texttt{zeroconf}, the weighted approaches \texttt{W} and \texttt{GW} obtain the smallest core (or one of equal size).
In particular, on the \enquote{embedded} model, \texttt{W} and \texttt{GW} obtain a 50\% reduction while the other approaches do not identify a non-trivial core.
This suggests that out of the considered heuristics, the weighted approach \texttt{W} / \texttt{GW} is the most efficient one on average (in terms of size).
These specific differences also suggest that our methods are sensitive to the particular heuristic, and developing further, more sophisticated methods can significantly improve performance.

While these results suggest that cores seem to be relatively rare in practice, we point out that this is mostly due to the specific structure of the benchmarks in the PRISM benchmark suite.
Some models are mostly non-deterministic, containing only a few probabilistic branches with \enquote{large} probabilities.
For example, \enquote{csma} describes a shared bus where the only source of probabilistic branching is choosing a random back-off delay after a collision---the smallest involved probability is in the order of $1/10$th.
Most models in the middle part of the table are of this structure.
On other classes, the parameters of the benchmark suite simply are not chosen large enough in order to obtain \enquote{unimportant} states.
This is the case with the \enquote{brp} model, where scaling up the maximal number of retransmissions can drastically improve the savings of our core approach, as we have shown in Section~\ref{sec:experiments}.
Finally, many models are infinitely repeating protocols and thus are strongly connected.
In light of Proposition~\ref{stm:reachable_mecs_in_core}, we cannot expect any improvement there.
All models of this kind are located in the bottom group of the table.

In general, most of the models in the benchmark suite describe abstract protocols.
For these, probabilistic branching is only present in few critical locations and is of a particular structure---mostly, randomness is used to, e.g., resolve ties or similar, resulting in rather large probabilities.
Here, our method is, by nature, not applicable.
However, as soon as real-life constraints are incorporated, many low-probability events are introduced to the model, for example hardware failures, sensor noise, or transmission errors due to environmental influences.
These low probability errors allow for non-trivial cores, as is the case with, for example, the \texttt{zeroconf} model.

\end{document}